\documentclass[journal]{IEEEtran}

\ifCLASSINFOpdf
\else
\fi

\IEEEoverridecommandlockouts                             

\usepackage{cite}
\usepackage{amsmath,amssymb,amsfonts}
\usepackage{algorithmic}
\usepackage{graphicx}
\usepackage{textcomp}
\usepackage{dsfont}
\usepackage{color}
\usepackage{epstopdf}        
\usepackage{enumerate}
\usepackage{lscape}
\usepackage{multicol}
\usepackage{multirow}
\usepackage{calligra}
\usepackage{booktabs}
\usepackage{mathtools}
\usepackage{framed} 
\usepackage{picins}
\usepackage{empheq}
\usepackage{graphics} 
\usepackage{epsfig} 
\usepackage{times} 
\usepackage{amsmath} 
\usepackage{amssymb}  
\usepackage{amsfonts}  
\usepackage{epstopdf}
\usepackage{enumerate}
\usepackage{graphicx} 
\usepackage{xcolor}
\usepackage{comment} 

\usepackage{pst-all}
\usepackage{pstricks}

\usepackage{subfigure} 

\usepackage{caption2} 

\usepackage{etoolbox}

\usepackage{soul}%
\usepackage{cancel}

\newtheorem{theorem}{Theorem}

\newtheorem{lemma}{Lemma}
\newtheorem{corollary}{Corollary}

\newtheorem{remark}{Remark}

\newtheorem{claim}{Claim}

\def\BibTeX{{\rm B\kern-.05em{\sc i\kern-.025em b}\kern-.08em
    T\kern-.1667em\lower.7ex\hbox{E}\kern-.125emX}}

\definecolor{Nicolas}{rgb}{0.017,0.59,0.55}
\newenvironment{SUBENVNicolas}{\color{Nicolas}}{\color{black}}

\definecolor{Jorge}{rgb}{0.217,0.51,0.255}
\newenvironment{SUBENVJorge}{\color{Jorge}}{\color{black}}

\definecolor{Miroslav}{rgb}{0.74,0.2,0.64}
\newenvironment{SUBENVMiroslav}{\color{Miroslav}}{\color{black}}

\begin{document}
\title{
\bf Prescribed-Time Newton Extremum Seeking\\ using Delays and Time-Periodic Gains}

\author{Nicolas Espitia,$^{1}$  Jorge I. Poveda$^{2}$ and  Miroslav Krstic.$^{3}$   
\thanks{This work was partially
supported by the Agence Nationale de la Recherche (ANR)
via grant PH-DIPSY ANR-24-CE48-1712, and by NSF grants CMMI-2228791 and ECCS-2210315.\\ $^{1}${\scriptsize CRIStAL UMR 9189 CNRS} - Centre de Recherche en Informatique Signal et Automatique de Lille - CNRS, Centrale Lille, Univ. Lille, F-59000 Lille, France. ({\tt\small e-mail: nicolas.espitia-hoyos@univ-lille.fr})}%
\thanks{$^{2}$Electrical and Computer Engineering Department,University of California, San Diego, USA.}
\thanks{$^{3}$Department of Mechanical and Aerospace Engineering, University of California, San Diego, USA.}

}

\maketitle

\begin{abstract}
We study \textit{prescribed-time} extremum seeking (PT-ES)  for scalar maps in the presence of time delays. The PT-ES problem has been studied by Yilmaz and Krstic in 2023 using chirpy probing and time-varying gains that grow unbounded.  
To alleviate the gain singularity, in this paper we present an alternative approach, employing delays with \emph{bounded time-periodic} gains, for achieving prescribed-time convergence to the extremum. Our results are not extensions or refinements of earlier works, but a new methodological direction --applicable even when the map has no delay. The main PT-ES algorithm  compensates the map's delay and uses  perturbation-based and the Newton (rather than gradient) approaches.  
With the help of averaging theorems in infinite dimension, specifically Retarded Functional Differential Equations (RFDEs),  we conduct a prescribed-time convergence analysis on a suitable averaged \textit{target} ES system, which contains the time-periodic gains of the map and feedback delays. We further extend our method to multivariable static maps and
illustrate our results through numerical simulations.

\end{abstract}
%

%

\section{Introduction}

\subsection{Literature on extremum seeking (ES)}

Extremum Seeking (ES), first introduced in a paper by Leblanc in 1922 \cite{Leblanc1922}, has become one of the most popular adaptive optimization algorithms over the years. It has garnered significant interest \cite{ AriyurKrstic2003, SCHEINKER_survey,poveda2017framework} and has led to numerous applications across various domains. This popularity stems from the algorithm's ability to systematically identify local optima in both static and dynamic nonlinear systems in real time, without requiring prior knowledge of these systems.

Many pioneering contributions can be highlighted in this field of research including: \cite{KrsticWang2000} which presents the first complete stability analysis for an ES algorithm, based on averaging and singular perturbation theories, for a general nonlinear dynamical system;  \cite{ChoiKrsticAriyurLee2002} which has extended the ES method to discrete dynamical systems; \cite{WangKrstic2000} which has applied ES to the problem of limit cycle maximization; \cite{TanNesicMareels2006} which achieved semi-global stability for ES algorithms; \cite{MoaseManzieBrear2010,GhaffariKrsticNesic2012} which have introduced  Newton-based ES algorithms to overcome the shortcomings of the Gradient-based ones;  \cite{LiuKrstic2012} which studied stochastic probing in ES; \cite{DurrStankovicEbenbauerJohansson2013,LabarGaroneKinnaertEbenbauer2019} studied Lie bracket approximation methods for ES systems; {\color{black}\cite{poveda2017framework,abdelgalil2023multi} which extended ES to systems with hybrid (continuous-time and discrete-time) dynamics}, and more recently \cite{FridmanZhang2020,ZhuFridmanOliveira2023} which have  employed a Time-Delay approach.

\subsection{Literature on ES with delays}

In all of these mentioned contributions, it is assumed that input/output delays are not present. However, in many processes feedback actions and output measurement may be subject to delays. These delays have to be taken into account in the ES algorithms, as they may induce instabilities or convergence property may no longer be guaranteed. The stabilization of ES systems is a challenging problem that requires using supplementary techniques to deal with the additional complexity introduced by the delay in both the design and the analysis. Only few results exist for this problem in ES, including \cite{OliveiraKrsticTsubakino2017}, in which the problem of extremum seeking for unknown static maps with constant pointwise delays, was formulated and solved for the first time using the backstepping approach for PDEs  and employing averaging in infinite-dimension \cite{Hale1990}.  Recently, this result was extended to the case of distributed delays in \cite{TsubakinoOliveiraKrstic2023}  and distributed optimization and
Nash equilibrium seeking in systems with delays  \cite{Oliveira_NashEquilibrium2021}. Since those pioneering contributions, some relevant  extensions of ES with e.g., time-varying and state-dependent delays, and furthermore with  other  PDE dynamics,  have been proposed and can be found in  \cite{Oliveira_bookESDelays}.   
 Another recent improvement of \cite{OliveiraKrsticTsubakino2017} is found in  \cite{FerreiraOliveiraKrstic2023} in which not only the optimization of the quadratic maps with constant pointwise delays is achieved, but also the transient towards the optimum itself is optimized. Other approaches include the use of sequential predictors \cite{MalisoffKrstic2021}.

Despite all this stream of results, most of them only consider asymptotic or exponential convergence guarantees towards the optimum, yet in several applications of ES (e.g.,   Air flow control \cite{KomatsuMiyamotoOhmoriSano2001,ChangMoura2009}, Antilock Braking Systems \cite{ZhangOrdonez2012}, Fusion power plant \cite{OuXuSchusterLuceFerronWalkerHumphreys2008},  process and reaction systems \cite{Dochain2011}, Nash Equilibrium Seeking\cite{FrihaufKrsticBasar2011}, \cite{Oliveira_NashEquilibrium2021}, Stock trading \cite{FormentinPrevidiMaroniCantaro2018},  Source Seeking \cite{Suttner2022}, Traffic Congestion Control \cite{Yu_traffic_ExtremumSeeking2021}
, etc.),   where   seeking  must occur within a finite given time, and delays need to be perfectly compensated, the \textit{finite-in-time} forms of convergence are strongly desired. 

\subsection{Literature on finite-in-time ES}
`Finite-in-time' convergence refers to several properties, with three standing out:  i) finite-time convergence \cite{BhatBernstein1995}, where the settling time is a finite but possibly depends on the system's initial conditions and parameters; ii) fixed-time convergence \cite{Polyakov2012},  where the settling time   is uniformly bounded by a constant that is independent of the initial conditions but may depend on system's parameters which can be tuned in some cases to achieve convergence in arbitrarily chosen finite time; and iii) prescribed-time (PT) convergence {\cite{SongWangHollowayKrstic2017,ochoa2023prescribed}} where the settling time is independent of both the initial conditions and the system's parameters and is usually assigned via time-varying singular gains. 

In addition to achieving prescribed-time convergence through time-varying singular gains, alternative methods have been proposed, such as attaining prescribed-time convergence via artificial delays combined with time-periodic gains\cite{Karafyllis2006}, \cite{Insperger2006}  and \cite{ZhouMichielsChen2022}--- an approach that leverages a key technical result presented in \cite[Chapter 3, pp 87-88]{Hale1993}, which  has been the focus of subsequent research efforts on stabilization in fixed time by means of periodic time-varying feedback with notable extensions  to observer-based and dual/output  feedback designs,  robustness to external additive disturbance,  uncertainties \cite{Zhou_PDF_Automatica2024}, as well as the integration  with sliding mode controllers for disturbance rejection \cite{Deng_Moulay2024}.

In the framework of `finite-in-time extremum seeking',  finite-time (semiglobal practical) ES was studied in \cite[Sec. 6.1]{poveda2017framework} and \cite{Guay2021}. In addition, Fixed-time ES was studied  in  \cite{PovedaKrstic2021} via non-smooth dynamics, with extensions to multi-agent systems \cite{Poveda2022}. Prescribed-time source seeking for static nonlinear maps was introduced in \cite{TodorovskiKrstic2023} and improvements to obtain unbiased extremum seeking were reported in \cite{Yilmaz_unbias2023}. 
 These results have been extended to the case of prescribed-time ES   unbiased extremum seeking with (parablic) PDE  actuators  \cite{Yilmaz_unbias_PDEs2024} and  for prescribed-time ES with  delay  in \cite[Section III]{CemalTugrul2024Tac} by relying on predictor-feedback whose design builds on the PDE backstpeping approach with time-varying backstepping Volterra transformation (as in \cite{EspitiaPerruquetti2022}); and by using  chirpy probing. A   novel 	averaging theorem in  infinite dimension  is proposed  in  \cite{CemalTugrul2024Tac}  to cope with the challenges of handling unbounded terms  due to the time-varying singular gain, and the lack of periodicity of the chirpy probing signals, among others.  

\subsection{Contributions}

To address the singular gain issue, in this paper, we present  a Newton-based extremum seeking (ES) algorithm that achieves prescribed-time convergence for  static maps with delayed outputs. Unlike \cite{CemalTugrul2024Tac}, our method does not employ time-varying singular gains or chirp-based probing (i.e., perturbation/demodulation signals with unboundedly increasing frequency), nor does it rely on backstepping techniques for PDEs.
Instead, we build on the intentional use of artificial delays and periodic time-varying gains that do not grow unbounded. This combination not only compensates for the output delay but also ensures  local convergence in a prescribed time. Furthermore, the algorithm remains well-defined beyond this prescribed time, thereby avoiding issues related to singularities.  Although the algorithm can operate after the prescribed finite time, theoretical guarantees are established only over arbitrarily large compact time intervals. It is worth mentioning that the intentional use of delays in control design has proved to be powerful in achieving objectives that usually  cannot be achieved using static  linear feedbacks  \cite{Niculescu2001,Fridman_Selivanov_2025}. Such objectives include for instance finite time convergence \cite{Karafyllis2006}, \cite{ZhouMichielsChen2022}, or disturbance-robust adaptive control \cite{Karafyllis2023Adaptive}.

 The main idea of our approach is indeed to combine periodic time-varying feedback with delays. In this sense, the  Newton-based scheme is critical: we propose a ``Riccati-like" filter that can estimate the inverse of the Hessian in prescribed time (in the average sense), enabling both  the learning dynamics and the estimator for the gradient to converge in a prescribed time, despite  the presence of the output delay that needs to be compensated.  To the best of our  knowledge, this structure is novel in the context of ES.
 
Since the notion of ``dead-beat'' Lyapunov stabilization by time-periodic delayed feedback originates, 
from Hale and Verduyn-Lunel \cite[Chapter 3, pp 87-88]{Hale1993} (the analysis; {1993}) and Karafyllis \cite{Karafyllis2006} (feedback design; 2006), we refer to our approach to prescribed-time ES as the ``Karafyllis, Hale, Verduyn-Lunel" (KHV) PT-ES approach. Insperger's \cite{Insperger2006} ``act-and-wait'' idea and \cite{Karafyllis2006} appeared in the same year. 
While both works offer valuable, distinct contributions in their respective domains, \cite{Karafyllis2006} provides a broader framework, including Lyapunov estimates and applicability to nonlinear systems, and being, in fact, quite a mathematical triumph. 
 We also acknowledge excellent subsequent efforts for linear systems, particularly \cite{ZhouMichielsChen2022}, and also note that the name suggested for the approach in \cite{ZhouMichielsChen2022} is a deft choice---``periodic delayed feedback'' (PDF).  We refer, however,  to our variety of PT-ES algorithms as {\em KHV PT-ES}. 

\subsection{Organization}
This paper is organized as follows. In Section \ref{Section:Problem_Statement_ExTS}, we provide some preliminaries on delay-compensated extremum seeking, the motivation and  a key idea on prescribed-time stabilization by means of time-varying delayed feedbacks. In section \ref{KHV_ES_algorithm} we present the prescribed-time Newton-based extremum seeking problem for a real-valued static map featuring an output delay.  In Section \ref{proof_of_main_result} we present  the main result and the extension to multivariable static maps, and further discussion. In Section \ref{numerical_simulations},  we consider a numerical example to illustrate the main results. Finally, we give some conclusions and perspectives in Section \ref{Section:Conclusion_ExTS}.\\

\textbf{Notation:}

 $\mathbb{R}_+$ denotes the set of non-negative real numbers.
 For a vector $x \in \mathbb{R}^{n}$ we denote $\vert x \vert$ its Euclidean norm. $\lceil x\rceil$ denotes the ceiling function,  i.e., $\lceil x\rceil = \min\{n\in \mathbb{Z}: n \geq x \}$.
 We denote by the unit circle centered around the origin the set $\mathbb{S}^1=\left\{u=(u_1,u_2)^{\top} \in \mathbb{R}^2: u_1^2+u_2^2=1\right\}$. Let $I \subseteq \mathbb{R}$ be a non-empty interval and $\Omega \subset \mathbb{R}^n$ be a non-empty set. By $C^0(I;\Omega)$ we denote  the class of continuous functions on $I$, which takes values in $\Omega$.  For $x \in C^{0}([-r,0];\mathbb{R}^n)$, we define $\Vert x \Vert=\max_{-r \leq s \leq 0}(\vert x(s) \vert)$
 Let $x:[a-r,b) \rightarrow \mathbb{R}^n$ with $b>a \geq 0$ and $r >0$. By $x_t$ we denote the ``history" of $x$ from $t-r$ to $t$, i.e.,  $x_t(s)=x(t+s)$; $s\in[-r,0]$, for $t\in [a,b)$. The class of functions $\mathcal{K}_{\infty}$ is the class of strictly increasing, continuous functions $\alpha: \mathbb{R}_+ \rightarrow \mathbb{R}_{+}$ with $\alpha(0)=0$ and $\lim_{s\rightarrow + \infty} \alpha(s)= + \infty$.

\section{Preliminaries and problem description }\label{Section:Problem_Statement_ExTS}

Consider the following real-valued output:
\begin{equation}\label{outputmap}
y(t)=Q(\theta(t-D)),
\end{equation}
 where $Q(\cdot)$ is an unknown mapping, and $D\geq0$ is a nonnegative and known constant delay.     The goal is to adjust the input $\theta \in \mathbb{R}$ in real time,  in order to drive and maintain the output $y(t)$ around the unknown extremum $y^{\star} =Q(\theta^\star)$, where  $\theta^\star \in \mathbb{R}$ is the unknown optimizer.  For  maximum seeking purposes, we assume  that 
\begin{equation}\label{eq:Gradient_Hessian_equilibrium}
 Q^{'}(\theta^\star)=0,  \quad  Q^{''}(\theta^\star)=H^\star <0,
\end{equation}
where the Hessian $H^\star$, of the static map,  is unknown.  The output $y$ is assumed to be quadratic\footnote{We recall that many
objective functions can be locally (in a small neighborhood of the extremum)
well approximated by a quadratic function.\\
 Finding the minimum or maximum  of  quadratic maps with or without known delay  can be done  via \textit{off-line  identification schemes}  e.g.,  interpolation with a collection  of input/output data points (three data points in the case of scalar maps). This problem can also be solved by  an  open-loop perturbation-based approach  with Hessian and gradient estimate on the initial time interval, as remarked in \cite[Remark 1]{Jbara2025}. See also  \cite{Labar2024}.\\
 In this paper we pursue \textit{real-time optimization} and chose to start with a scalar quadratic map  in order to facilitate understanding and convey our new methodology pedagogically, while
avoiding overly complex notation and generalizations at this stage. An extension of our method to  multivariable ES for static maps can be found in Section \ref{Multivariable_ES}. }. Therefore, it can be written as 
\begin{equation}\label{eq:static_map_without_perturbation}
y(t)=  y^\star + \tfrac{H^\star}{2}\left(\theta(t-D) - \theta^\star \right)^{2}.
\end{equation}

\subsection{Preliminaries on delay-compensated Newton ES with exponential convergence} \label{Preliminaries_delay_compensated_Tiago}

The actual input $\theta(t)=\hat{\theta}(t)+S_D(t)$ is a real estimate $\hat{\theta}(t)$ of $\theta^\star$ that is perturbed by $S_D(t)$ (to be specified below). Hence, \eqref{eq:static_map_without_perturbation} is expressed as follows:
\begin{equation}\label{eq:outout_function_static_maps}
y(t)=  y^\star + \tfrac{H^\star}{2}\left(\hat{\theta}(t-D) +S_D(t-D) - \theta^\star \right)^{2}. 
\end{equation}
The following delay-compensated Newton-based extremum seeking  algorithm for static maps was  introduced \cite{OliveiraKrsticTsubakino2017} (particularized here for the scalar case, for the sake of clarity in the presentation):
\begin{align}
&\dot{\hat{\theta}}(t)=  z(t),  \label{eq:Newton_ES_static_maps1} \\
&\dot{z}(t) = - cz(t) -cK\Gamma(t)M(t)y(t)  - cK \int_{t-D}^{t}z(s)ds,  \label{eq:Newton_ES_static_maps2} \\
&\dot{\Gamma}(t)=w_r\Gamma(t) - w_rN(t)y(t)\Gamma^2(t),  \label{eq:Newton_ES_static_maps3} 
\end{align}
where $\hat{\theta}(t)\in\mathbb{R}$ is the learning dynamics (an estimator of  $\theta^\star$), $z(t)\in\mathbb{R}$ is the state of a filter aiming at estimating the gradient, with $c>K>0$,  and $\Gamma(t)\in\mathbb{R}$ is the state of a Riccati filter (Riccati differential equation) aiming at estimating the inverse of the Hessian with $w_r>0$  a positive real tunable parameter.    The Dither signals $S_D(t)$, $M(t)$, and   $N(t)$ are given as follows:
\begin{align}
S_D(t)&=a\sin(\omega(t+D)), \label{eq:dither_signal_as_tiago_S}\\ 
M(t) &=\tfrac{2}{a}\sin(\omega t),\label{eq:dither_signal_as_tiago_M}\\ 
N(t)&=\tfrac{16}{a^2}\left(\sin^2(\omega t) - \tfrac{1}{2} \right). \label{eq:dither_signal_as_tiago_N}
\end{align} 
where $0 < a < 1$, and $\omega>0$. 
 Defining, for a given constant $y^\star \in \mathbb{R}$, $y_c(t,  \chi)=y^\star +\frac{H^\star}{2}\left(S_D(t-D) +  \chi \right)^{2}$ and $\Pi= \tfrac{2\pi}{\omega}$,  we recall  the following averaging   properties (see e.g. \cite{GhaffariKrsticNesic2012}, \cite{OliveiraKrsticTsubakino2017} or  \cite[Lemma 1]{MalisoffKrstic2021}):
\begin{align}
\tfrac{1}{\Pi}\int^{\Pi}_{0} M(s)y_c(s,\chi)ds &= H^\star  \chi,\label{property:averaging_get_gradient}\\
\tfrac{1}{\Pi}\int^{\Pi}_{0} N(s)y_c(s,\chi)ds &= H^\star. \label{property:averaging_get_Hessian}
\end{align}
The use of the Riccati filter in \eqref{eq:Newton_ES_static_maps1}-\eqref{eq:Newton_ES_static_maps3} is crucial and is the main feature of the Newton-based Extremum seeking algorithm that was originally introduced in \cite{Nesic2010}  and subsequently in \cite{GhaffariKrsticNesic2012} for the multivariable case. Under this algorithm, one can obtain faster convergence  (with a convergence rate independent of the Hessian) compared to a traditional Gradient-based Extremum seeking algorithm \cite[Section IV]{OliveiraKrsticTsubakino2017}. In the multivariable case, the use of the Riccati filter is particularly crucial to avoid having to invert directly the matrix of the Hessian estimate.

We recall the main ideas behind the delay-compensated Newton-based extremum seeking algorithm \eqref{eq:Newton_ES_static_maps1}-\eqref{eq:Newton_ES_static_maps3} introduced in \cite{OliveiraKrsticTsubakino2017}.  It  allows to compensate for the delay thanks to the distributed term in \eqref{eq:Newton_ES_static_maps2} i.e.,  using the history of the state of the filter $z(t)$ on a interval $[t-D,t]$. Averaging the error dynamics,  reformulating the problem as an input delay control, and then applying the backstepping method for PDEs allow to come up with a predictor-based controller and, therefore, the updating law for the learning dynamics.  The stability analysis is carried out on the averaged error dynamics, followed by  Lyapunov analysis. This methodology enables to invoke the averaging theorem in infinite-dimension (see \cite{Hale1990}).
The method carries out not only scalar static maps but also multivariable maps with distinct delays. 
Using one-stage sub-predictors, an alternative approach to compensate for the delays was proposed in  \cite{MalisoffKrstic2021}. The algorithm does not involve distributed terms but pointwise delays (see \cite[Section 6]{MalisoffKrstic2021} for a Newton-based ES).  The stability analysis of the averaged error dynamics is performed using   Lyapunov–Krasovskii techniques,  yielding a sufficient condition on the delay  and the convergence rate of the algorithm (slower convergence for larger delays).   The existing averaging theory for infinite dimensional systems applies as well.  

\subsection{Discussion of Newton ES  algorithm  \eqref{eq:Newton_ES_static_maps1}-\eqref{eq:Newton_ES_static_maps3}}\label{Discussion_on_Tiago_approach}
Notice that,  by means of the Fundamental Theorem of Calculus, the Newton-based Extremum seeking algorithm \eqref{eq:Newton_ES_static_maps1}-\eqref{eq:Newton_ES_static_maps3}, with static map \eqref{eq:outout_function_static_maps} can be reformulated as follows:
\begin{align} 
&\dot{\hat{\theta}}(t)=  z(t), \label{eq:Newton_ES_static_maps1V2} \\
&\dot{z}(t) = - cz(t) -cK\Gamma(t)M(t)y(t)  - cK \hat{\theta}(t) + cK \hat{\theta}(t-D),  \label{eq:Newton_ES_static_maps2V2} \\
&\dot{\Gamma}(t)=w_r\Gamma(t) - w_rN(t)y(t)\Gamma^2(t),  \label{eq:Newton_ES_static_maps3V2}  
\end{align} 
where we can clearly see the presence of a ``state" delay of the learning dynamics in both \eqref{eq:Newton_ES_static_maps2V2} and the output \eqref{eq:outout_function_static_maps}. This reformulation may look similar to the algorithm using  one-stage sequential predictors in \cite{MalisoffKrstic2021}; with the difference, that in \eqref{eq:Newton_ES_static_maps1V2}-\eqref{eq:Newton_ES_static_maps3V2}  the delay is only being present in  the state of the learning dynamics  and one does not use approximations of the predictor, but an exact expression via Fundamental Theorem of Calculus.
If one applies averaging to the error dynamics (with change of variables $\tilde{\theta}=\hat{\theta}- \theta^\star$,  $\tilde{\Gamma}=\Gamma- H^{\star-1}$)  by using properties \eqref{property:averaging_get_gradient}-\eqref{property:averaging_get_Hessian} in conjunction with the averaging operation  \cite{Hale1990}, the resulting averaged error dynamics  will read as follows:
\begin{align}
&\dot{\tilde{\theta}}_{\rm av}(t)=  z_{\rm av}(t),  \label{eq:Newton_ES_static_averaged1} \\
&\dot{z}_{\rm av}(t) = - cz_{\rm av}(t) -cK\tilde{\theta}_{\rm av}(t)  - cK \tilde{\Gamma}_{\rm av}(t)H^{\star} \tilde{\theta}_{\rm av}(t-D),  \label{eq:Newton_ES_static_averaged2} \\
&\dot{\tilde{\Gamma}}_{\rm av}(t)=-w_r\tilde{\Gamma}_{\rm av}(t) - w_rH^{\star}\tilde{\Gamma}_{\rm av}^2(t),  \label{eq:Newton_ES_static_averaged3} 
\end{align} 
where we can observe that \eqref{eq:Newton_ES_static_averaged1}-\eqref{eq:Newton_ES_static_averaged2} is cascaded by $\tilde{\Gamma}_{\rm av}(t)$.  In fact, the  multiplicative term $ \tilde{\Gamma}_{\rm av}(t)H^{\star} \tilde{\theta}_{\rm av}(t-D) $ in \eqref{eq:Newton_ES_static_averaged2} can be seen as a multiplicative perturbation that eventually goes to zero, since $\tilde{\Gamma}_{\rm av}(t) \rightarrow 0$ (fast enough depending on the choice of $w_r$) as  $t\rightarrow \infty$, and  as long as  $\vert  \tilde{\theta}_{\rm av} (t-D)\vert$ remains bounded. The stability result then would follow if,  instead of employing the PDE backstepping approach as in \cite{OliveiraKrsticTsubakino2017},   one were to  apply the classical theory of time-delay systems  as  demonstrated  in \cite{MalisoffKrstic2021} (possibly with similar trade-off of slow convergence rate for arbitrarily long delays).  Furthermore, the error dynamics from \eqref{eq:Newton_ES_static_maps1V2}-\eqref{eq:Newton_ES_static_maps3V2} and  the averaged error dynamics  \eqref{eq:Newton_ES_static_averaged1}-\eqref{eq:Newton_ES_static_averaged3} could be represented as Retarded Functional Differential Equations (RFDEs) for which  classical theorems on averaging in infinite dimensions  are also applicable. These insights   may further justify, at intuitively level,  why  the original algorithm (in \cite{OliveiraKrsticTsubakino2017}) \eqref{eq:Newton_ES_static_maps1}-\eqref{eq:Newton_ES_static_maps3} is capable of actually estimating  the optimal point, hence maximizing the static map, even in the presence of delays.

\subsection{Towards KHV approach to prescribed-time ES}

The  discussion in Subsection \ref{Discussion_on_Tiago_approach}, together with the observation leading to the ES algorithm \eqref{eq:Newton_ES_static_maps1V2}-\eqref{eq:Newton_ES_static_maps3V2} provides an intuitive understanding of the algorithm's operation. It also motivates the design of new ES algorithms that not only compensate for map delays but also accelerate convergence, potentially achieving finite-time convergence.
  In particular, the  trade-off identified  in \cite{MalisoffKrstic2021} between slower convergence versus arbitrarily large delays,   further underscores the seek for accelerated ES algorithms. More specifically, the seek for algorithms  enabling convergence within  prescribed finite time without relying necessarily on the use of chirpy probing and   time-varying singular gains (as in \cite{CemalTugrul2024Tac}), while ensuring solutions  exist beyond post-terminal times. 
  
  This objective  can be reframed as requiring  the averaged  error dynamics   $\tilde{\theta}_{\rm av}(t)$ to converge to zero in a prescribed time. This goal is indeed attainable. To see this,  notice first that the averaged error dynamics \eqref{eq:Newton_ES_static_averaged1}-\eqref{eq:Newton_ES_static_averaged2} is cascaded by $\tilde{\Gamma}_{\rm av}(t)$ and  that the   term $ \tilde{\Gamma}_{\rm av}(t)H^{\star} \tilde{\theta}_{\rm av}(t-D) $ appearing in \eqref{eq:Newton_ES_static_averaged2} can be seen as a perturbation. Suppose for the time being,  that  this perturbation  can be made to vanish in finite time e.g.,  by setting $\tilde{\Gamma}_{\rm av}(t) = 0$, for all $t\geq T^{'}$ (with $T^{'}>D$). Then the states of the resulting closed-loop averaged error dynamics $(\tilde{\theta}_{\rm av}(t),z_{\rm av}(t))^{\top}$ would converge to zero exponentially, for all $t\geq T^{'}$. Nevertheless, our goal is to achieve  a more demanding type of convergence: prescribed-time convergence. Specifically, we want   $(\tilde{\theta}_{\rm av}(t),z_{\rm av}(t))^{\top}$ to reach zero in a prescribed-time and remain at zero thereafter, together with a prescribed-time convergent Riccati-like filter for inverse Hessian estimation.  This calls for a new approach that combines i) time-varying tools (still falling within the framework of \textit{prescribed-time} control), ii) the intentional use of delays; drawing upon  the core idea that we  present   next.

\subsubsection{Key idea: PT stabilization of single integrator by time-periodic  delayed  feedback}\label{section_key_idea}

Consider the   following  one-dimensional  system operating on $[0,3T]$: 
\begin{equation}\label{eq:single_integrator_in_coreIDEA}
\dot{x}(t)=u(t), \quad  \forall t \in [0, 3T], 
\end{equation}
 with arbitrary initial condition $x(s)=\phi(s)$, $\phi \in C^0([-T,0];\mathbb{R})$ and  the following $2 T$-periodic time-varying pointwise delayed feedback
\begin{equation}\label{eq:time-varying_delayed-feedback_key_result_in_coreIDEA}
u(t)= -\mathcal{K}_{T}\left(t\right)x(t-T),
\end{equation} 
where  $T>0$ is a prescribed number,   and $\mathcal{K}_T$ is  a $2T$-periodic pulse signal, given as follows: 
\begin{equation}\label{eq:mathcalK_T_time_sequence_in_coreIDEA}
\mathcal{K}_T(t)= \begin{cases}0&  \text{if}  \quad   t \in [0,T) \cup (2T,3T], \\
\tfrac{1}{T}\left( 1- \cos\left(\tfrac{2\pi t}{T} \right) \right) &\text{if}  \quad t \in [T,2T].
\end{cases}
\end{equation}  
Notice that the controller \eqref{eq:time-varying_delayed-feedback_key_result_in_coreIDEA} features: 1) zero gain on $[0,T] $ and keeps constant the state  $x(t)$; 2)  \textit{mean} gain $\tfrac{1}{T}$  over the second $T$-period, acting on the constant state from the first period using a delay, and resulting in dead-beat convergence at $t=2T$. Indeed, on $[T,2T]$ the control is given  by
$u(t)= -\tfrac{1}{T}x(0) + \tfrac{1}{T}\cos\left(\tfrac{2\pi t}{T} \right)x(0), $
 and the resulting closed-loop solution is given by $$x(t)=x(0) - \tfrac{t-T}{T}x(0) + \tfrac{1}{2\pi}\sin\left(\tfrac{2\pi t}{T} \right)x(0), $$
  from which one gets  $x(2T)=0$.  Finally, notice that the control features zero  gain on  the interval $(2T,3T) $,  hence, $x(t)$ remains zero for $t \geq 2T$.
\begin{figure*}[t]
\centering{
\subfigure{\includegraphics[width=0.88\columnwidth]{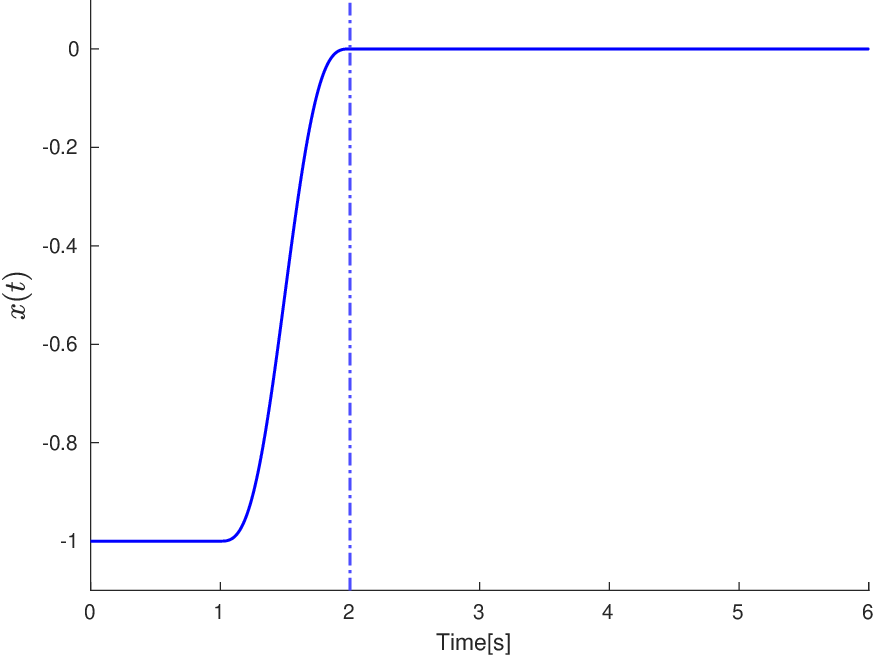} }
\subfigure{\includegraphics[width=0.88\columnwidth]{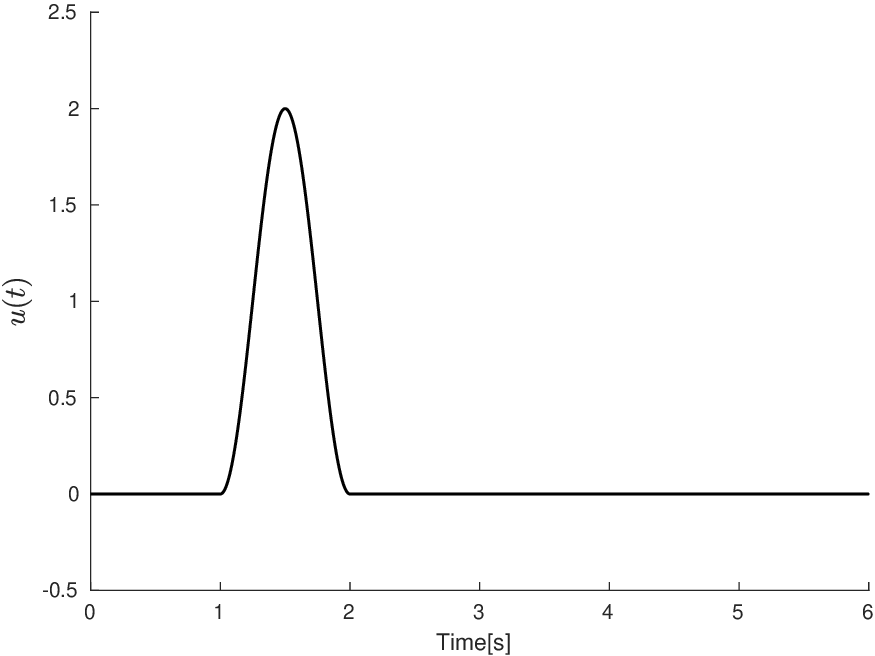} }
\caption{\underline{Left:} Solution of \eqref{eq:single_integrator_in_coreIDEA} with the $2 T$ - periodic time-varying \textit{pointwise} delayed feedback \eqref{eq:time-varying_delayed-feedback_key_result_in_coreIDEA}-\eqref{eq:mathcalK_T_time_sequence_in_coreIDEA}  with $T=1s$.  \underline{Right:} Profile of the $2 T$- periodic time-varying \textit{pointwise} delayed feedback \eqref{eq:time-varying_delayed-feedback_key_result_in_coreIDEA}-\eqref{eq:mathcalK_T_time_sequence_in_coreIDEA}.}
\label{closed-loop_solution_and_control_in_COREIDEA}
}
\end{figure*}
\begin{figure*}[t]
\centering{
\subfigure{\includegraphics[width=0.9\columnwidth]{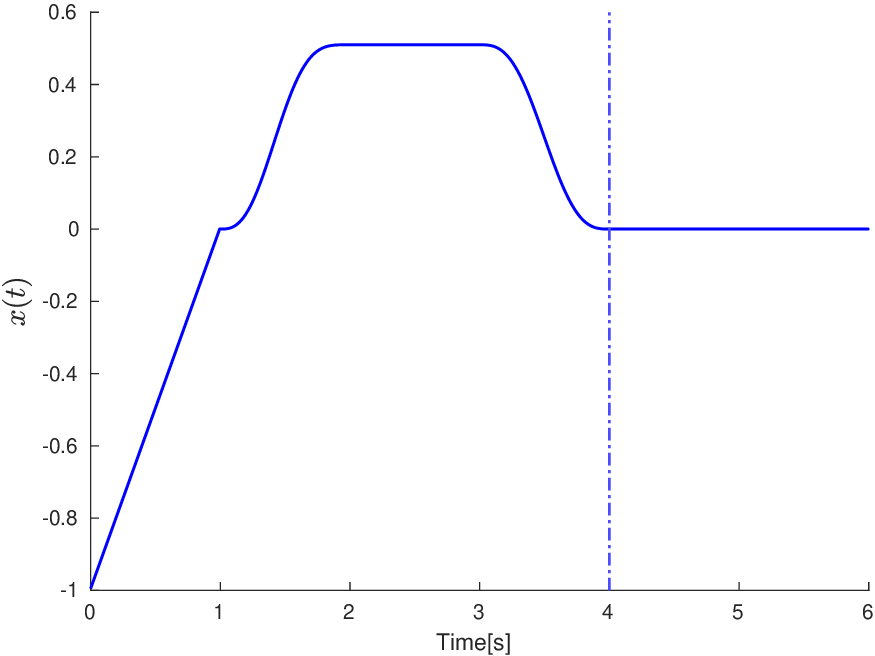} }
\subfigure{\includegraphics[width=0.9\columnwidth]{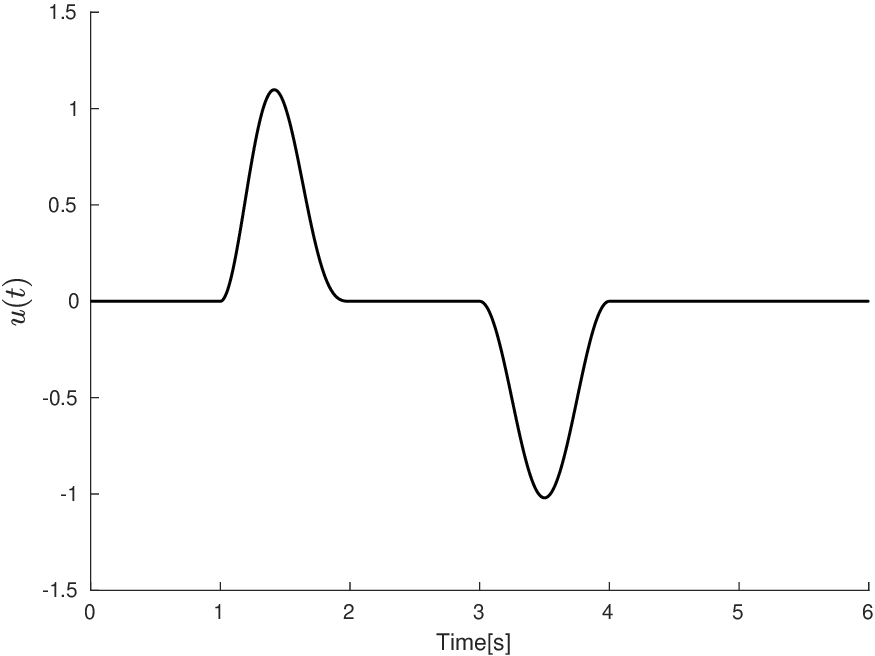} }
\caption{\underline{Left:} Solution of \eqref{eq:single_integrator_in_coreIDEA}    with the $2 T$ - periodic time-varying \textit{pointwise} delayed feedback \eqref{eq:time-varying_delayed-feedback_key_result_in_coreIDEA}-\eqref{eq:mathcalK_T_time_sequence_in_coreIDEA}  with $T=1s$, and which is subject  to an additive vanishing {\color{black}input} e.g., $v(t)=1$ for $t < 1s$, and $v(t)=0$ for $t \geq 1s$. \underline{Right:} Profile of the $2 T$- periodic time-varying \textit{pointwise} delayed feedback \eqref{eq:time-varying_delayed-feedback_key_result_in_coreIDEA}-\eqref{eq:mathcalK_T_time_sequence_in_coreIDEA}.}
\label{closed-loop_solution_and_control_in_COREIDEA_with_vanishing_distrubance}
}
\end{figure*}

Figure \ref{closed-loop_solution_and_control_in_COREIDEA} illustrates the dead-beat convergence property for the system \eqref{eq:single_integrator_in_coreIDEA},  with initial condition $x(0)=-1$ (and its preceding delay value), under the $2 T$-periodic time-varying pointwise delayed feedback \eqref{eq:time-varying_delayed-feedback_key_result_in_coreIDEA}-\eqref{eq:mathcalK_T_time_sequence_in_coreIDEA} with $T=1$s.   

 Suppose also that  system \eqref{eq:single_integrator_in_coreIDEA} is subject to an additive and   vanishing {\color{black}input}. Using the same controller \eqref{eq:time-varying_delayed-feedback_key_result_in_coreIDEA}-\eqref{eq:mathcalK_T_time_sequence_in_coreIDEA},  the  closed-loop system reads as  $\dot{x}(t)= -\mathcal{K}_{T}\left(t\right)x(t-T) + v(t)$, where $v(t) \in \mathbb{R}$, with $v(t) = 0$ for  $t\geq t^{'}$. In this scenario,  the dead-beat convergence is achieved at a multiple of $2T$.   
Figure \ref{closed-loop_solution_and_control_in_COREIDEA_with_vanishing_distrubance}  shows (on the left) the solution of the system \eqref{eq:single_integrator_in_coreIDEA},  with initial condition $x(0)=-1$ (and its preceding delay value), under the $2 T$-periodic time-varying pointwise delayed feedback \eqref{eq:time-varying_delayed-feedback_key_result_in_coreIDEA}-\eqref{eq:mathcalK_T_time_sequence_in_coreIDEA} with $T=1$s which is subject to a vanishing {\color{black}input}  e.g., $v(t)=1$ for $t < 1s$, and $v(t)=0$ for $t \geq 1s$. As it can be observed,   
the dead-beat convergence is achieved at a multiple of $2T$ (in this case, at $4$s).

\begin{remark}\label{remark_historical_issues0}
The dead-beat property discussed above for the same class of system, and time-varying gain (using the relation $\tfrac{2}{T}\sin^2\left(\tfrac{\pi t}{T} \right) = \tfrac{1}{T}\left( 1- \cos\left(\tfrac{2\pi t}{T} \right) \right)$),  can be found in, e.g.,  \cite[Chapter 3, pp 87-88]{Hale1993}; which  constitutes a key fundamental result  and that has been the focus of subsequent research efforts on stabilization in fixed time by means of periodic time-varying feedback systems for both linear \cite{Insperger2006,ZhouMichielsChen2022} and nonlinear systems \cite{Karafyllis2006}, \cite{Ding_stronPT-nonlinear2024}. 
\end{remark}

\subsubsection{Oscillators and time-periodic gains}

 Let us define first some signals of interest.  We consider the following  constrained linear oscillator  evolving on $\mathbb{S}^1$, for a given $T >0$:
\begin{equation}\label{eq:oscillator_system_periodic_time-varying_feedbacks}
\dot{\zeta}(t)= \mathcal{R}(\tfrac{\pi}{T}) \zeta(t), \quad \zeta=(\zeta_1,\zeta_2)^\top \in \mathbb{S}^1, 
\end{equation}
 where the skew symmetric matrix $\mathcal{R}(s) \in \mathbb{R}^{2\times 2}$ is defined as $ \mathcal{R}(s)=\left[\begin{array}{cc}
0 & s \\
-s & 0
\end{array}\right]$. 
\begin{figure*}[t]
\centering{
\subfigure{\includegraphics[width=0.9\columnwidth]{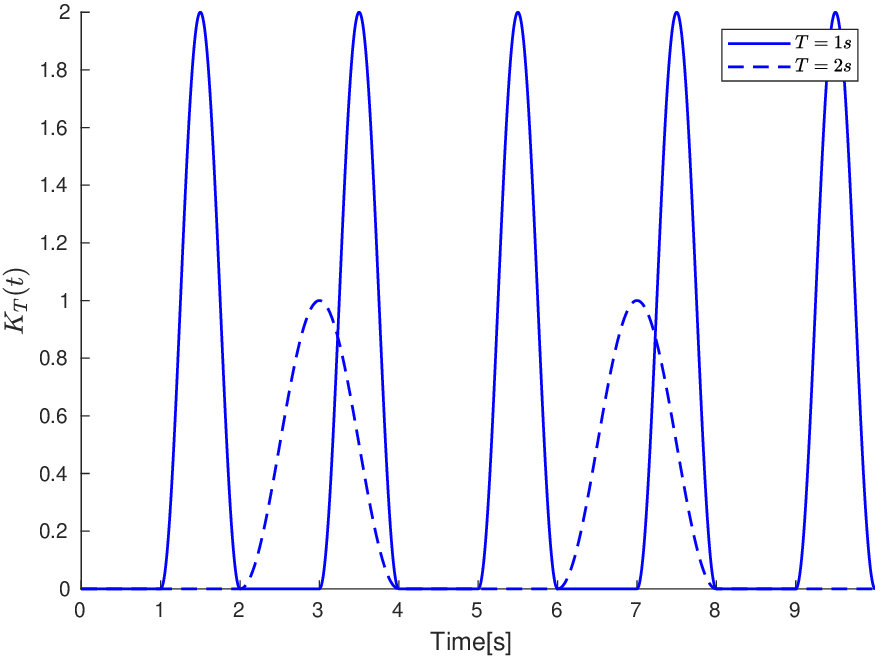} }
\subfigure{\includegraphics[width=0.9\columnwidth]{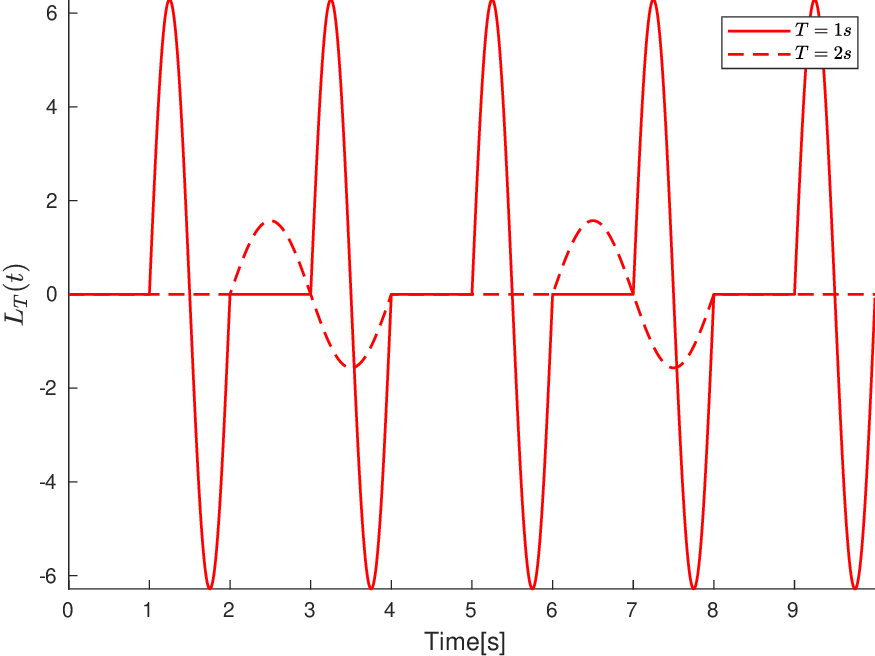} }
\caption{\underline{Left:} Profile of the $2T$- periodic signal $\mathcal{K}_T(t)$ defined in  \eqref{eq:mathcalK_T}  generated by the linear constrained oscillator \eqref{eq:oscillator_system_periodic_time-varying_feedbacks} with $T=1s$ (blue line) and $T=2s$ (blue dashed line),  and initial condition $\zeta(0)=(0,1)^{\top}\in \mathbb{S}^1$, which gives $\delta_{\zeta} = 0$, according to \eqref{eq:definition_of_delta-phase_of_oscillators}.  \underline{Right:} Profile of the $2T$- periodic signal $\mathcal{L}_T(t)$ defined in  \eqref{eq:mathcalL_T}  generated by the linear constrained oscillator \eqref{eq:oscillator_system_periodic_time-varying_feedbacks} with $T=1s$ (red line) and $T=2s$ (red dashed line),  and initial condition $\zeta(0)=(0,1)^{\top}\in \mathbb{S}^1$, which gives $\delta_{\zeta} = 0$, according to \eqref{eq:definition_of_delta-phase_of_oscillators}. }
\label{Pulse_signals}
}
\end{figure*}
The explicit solutions to \eqref{eq:oscillator_system_periodic_time-varying_feedbacks} are given as follows:
\begin{align}
\zeta_1(t)&=\zeta_{1}(0) \cos \left( \tfrac{\pi}{T}t\right)+\zeta_{2}(0) \sin \left(\tfrac{\pi}{T}t\right), \label{eq:explicit_solution_osicillator1}\\
\zeta_2(t)&=-\zeta_{1}(0) \sin \left(\tfrac{\pi}{T}t\right)+\zeta_{2}(0) \cos \left(\tfrac{\pi}{T}t\right), \label{eq:explicit_solution_osicillator2}
\end{align}
for any initial conditions $\zeta(0)=(\zeta_{1}(0),\zeta_{2}(0))^\top\in \mathbb{S}^1$. Moreover, \eqref{eq:explicit_solution_osicillator1}-\eqref{eq:explicit_solution_osicillator2} can further be expressed as
\begin{align}
\zeta_1(t)&= \sin \left(\tfrac{\pi}{T}\left( t + \delta_{\zeta} T \right)\right), \label{eq:explicit_solution_osicillator1_V2}\\
\zeta_2(t)&=\cos \left(\tfrac{\pi}{T}\left( t + \delta_{\zeta} T \right)\right), \label{eq:explicit_solution_osicillator2_V2}
\end{align}
where $\delta_{\zeta} \in (-1,1]$ is given by
\begin{equation}\label{eq:definition_of_delta-phase_of_oscillators}
\delta_{\zeta} = \Delta(\zeta_1(0),\zeta_2(0)),  
\end{equation}
 and the function $\Delta(s_1,s_2)$  is defined as follows\footnote{It can be seen as a variation of the  2-argument arctangent (``atan2") function.}: 
\begin{equation}\label{eq:definition_variation_atan2_function}
\Delta(s_1,s_2) = 
\begin{cases} 
\frac{1}{\pi} \arctan\left(\frac{s_1}{s_2}\right) & \text{if } s_2 > 0, \\
\frac{1}{\pi} \arctan\left(\frac{s_1}{s_2}\right) + 1 & \text{if } s_2 < 0 \text{ and } s_1 \geq 0, \\
\frac{1}{\pi} \arctan\left(\frac{s_1}{s_2}\right) - 1 & \text{if } s_2 < 0 \text{ and } s_1 < 0, \\
\frac{1}{2} & \text{if } s_2 = 0 \text{ and } s_1 =1 , \\
-\frac{1}{2} & \text{if } s_2 = 0 \text{ and } s_1 =-1, \\
0 & \text{if } s_2 =1 \text{ and } s_1 = 0, \\
1 & \text{if } s_2 =-1\text{ and }s_1 = 0.
\end{cases}
\end{equation}
Based on \eqref{eq:explicit_solution_osicillator1_V2} we define the following $2T$-periodic pulse signal $\mathcal{K}_T(t)$  as a  smooth state-dependent  signal generated by the  constrained linear oscillator \eqref{eq:oscillator_system_periodic_time-varying_feedbacks}:
\begin{equation}\label{eq:mathcalK_T_V0}
\mathcal{K}_{T}(t)= \begin{cases}0&  \text{if}  \quad \zeta_1(t)>0, \\
\tfrac{2}{T}\zeta_1^2(t) &\text{if}  \quad \zeta_1(t)\leq 0.   \end{cases}
\end{equation} 
 which can, in turn, be written as follows:
\begin{equation}\label{eq:mathcalK_T}
\begin{split}
\mathcal{K}_{T}(t)= & \frac{2}{T}\max\big\{0, -\zeta_1(t) \vert\zeta_1(t) \vert \big\},~~\forall~t\geq0,
\end{split} 
\end{equation}
which verifies the following property:
\begin{equation}\label{eq:property_K_T_and delayed_K}
\mathcal{K}_{T}(t)\mathcal{K}_{T}(t-T)=0,~~\forall~t\geq 0.
\end{equation}
Moreover, we introduce another periodic pulse signal generated also by \eqref{eq:oscillator_system_periodic_time-varying_feedbacks}, corresponding to the time-derivative of $\mathcal{K}_{T}$. We denote it as $\mathcal{L}_{T}=\tfrac{d \mathcal{K}_T}{dt}$, which is given by:
\begin{equation}\label{eq:mathcalL_T}
\mathcal{L}_{T}(t)=  -\frac{4}{T^2}\pi\max\big\{0, -\zeta_1(t)  \big\} \zeta_2(t),~~\forall~t\geq0.
\end{equation}
Figure \ref{Pulse_signals} depicts an example of the profiles of  the $2T$-periodic signals $\mathcal{K}_T$ (left - blue lines) and $\mathcal{L}_{T}$ (right - red lines) 
 with $T=1s$ and $T=2s$,   and initial conditions $\zeta(0)=(0,1)^{\top}\in \mathbb{S}^1$, which gives $\delta_{\zeta} = 0$, according to \eqref{eq:definition_of_delta-phase_of_oscillators}-\eqref{eq:definition_variation_atan2_function}.  

Finally, we consider the following constrained linear  oscillator  evolving on $\mathbb{S}^1$: 
\begin{equation}\label{eq:oscillator_system_periodic_time-varying_feedbacks_tau}
\dot{\xi}(t)= \mathcal{R}(\tfrac{\pi}{2T}) \xi(t), \quad \xi=(\xi_1,\xi_2)^\top \in \mathbb{S}^1,   
\end{equation}
where $T$ is the same as in \eqref{eq:oscillator_system_periodic_time-varying_feedbacks}, and which satisfies
\begin{align}
\xi_1(t)&= \sin \left(\tfrac{\pi}{2T}\left( t + \delta_{\xi} 2T \right)\right), \label{eq:explicit_solution_osicillatortau_1_V2}\\
\xi_2(t)&=\cos \left(\tfrac{\pi}{2T}\left( t + \delta_{\xi} 2T \right)\right), \label{eq:explicit_solution_osicillatortau_2_V2}
\end{align}
where $\delta_{\xi} \in (-1,1]$ is characterized  as follows: 
\begin{equation}\label{eq:definition_of_delta-xi-phase_of_oscillators}
\delta_{\xi}=\Delta(\xi_1(0),\xi_2(0)),
\end{equation}
with $\Delta$ being defined in \eqref{eq:definition_variation_atan2_function}. 

Finally, we define the following pulse signal of interest:
\begin{equation}\label{eq:mathcalK_tau}
\begin{split}
\mathcal{K}_{ 2T}(t)= & \frac{1}{ T}\max\big\{0, -\xi_1(t) \vert\xi_1(t) \vert \big\},~~\forall~t\geq0.
\end{split} 
\end{equation} 
\begin{remark}\label{remk:forward_invariance_oscillators}
Let us remark first  that,  by the structure of the oscillators  \eqref{eq:oscillator_system_periodic_time-varying_feedbacks} and \eqref{eq:oscillator_system_periodic_time-varying_feedbacks_tau}, the set $\mathbb{S}^1$ is forward
invariant, i.e., if  $\zeta(0) \in \mathbb{S}^1$, and $\xi(0) \in \mathbb{S}^1$   then  $\zeta(t) \in \mathbb{S}^1$ and $\xi(t) \in \mathbb{S}^1$  for all $t \geq 0$.
Moreover, since in our formulation no solutions of systems \eqref{eq:oscillator_system_periodic_time-varying_feedbacks} and \eqref{eq:oscillator_system_periodic_time-varying_feedbacks_tau} are defined outside of the unit circle, the set $\mathbb{S}^1$ is trivially globally attractive. 
\end{remark}
\begin{remark}\label{remk:state-dependent-signals-time-invariantES}
Notice also that we have considered  the time-varying periodic signals \eqref{eq:mathcalK_T}, \eqref{eq:mathcalL_T} and \eqref{eq:mathcalK_tau}, as state-dependent signals -- those generated by the oscillators  \eqref{eq:oscillator_system_periodic_time-varying_feedbacks} and \eqref{eq:oscillator_system_periodic_time-varying_feedbacks_tau}, respectively. Those signals will be part of the dynamics of the ES algorithm, and they will allow us  
to treat its averaged error-dynamics  as a time-invariant RFDE. This approach significantly facilitates the application of classical averaging theorems in infinite-dimensions.
\end{remark}

\section{KHV Prescribed-time Newton  extremum seeking}\label{KHV_ES_algorithm}

 We are now in position to  propose the following Prescribed-time Newton-based  extremum seeking by means of time-varying pointwise delayed feedback, which we will often refer to as \textit{KHV PT-ES}:
\begin{figure*}[t!]
\centering
  \includegraphics[width=0.7\textwidth]{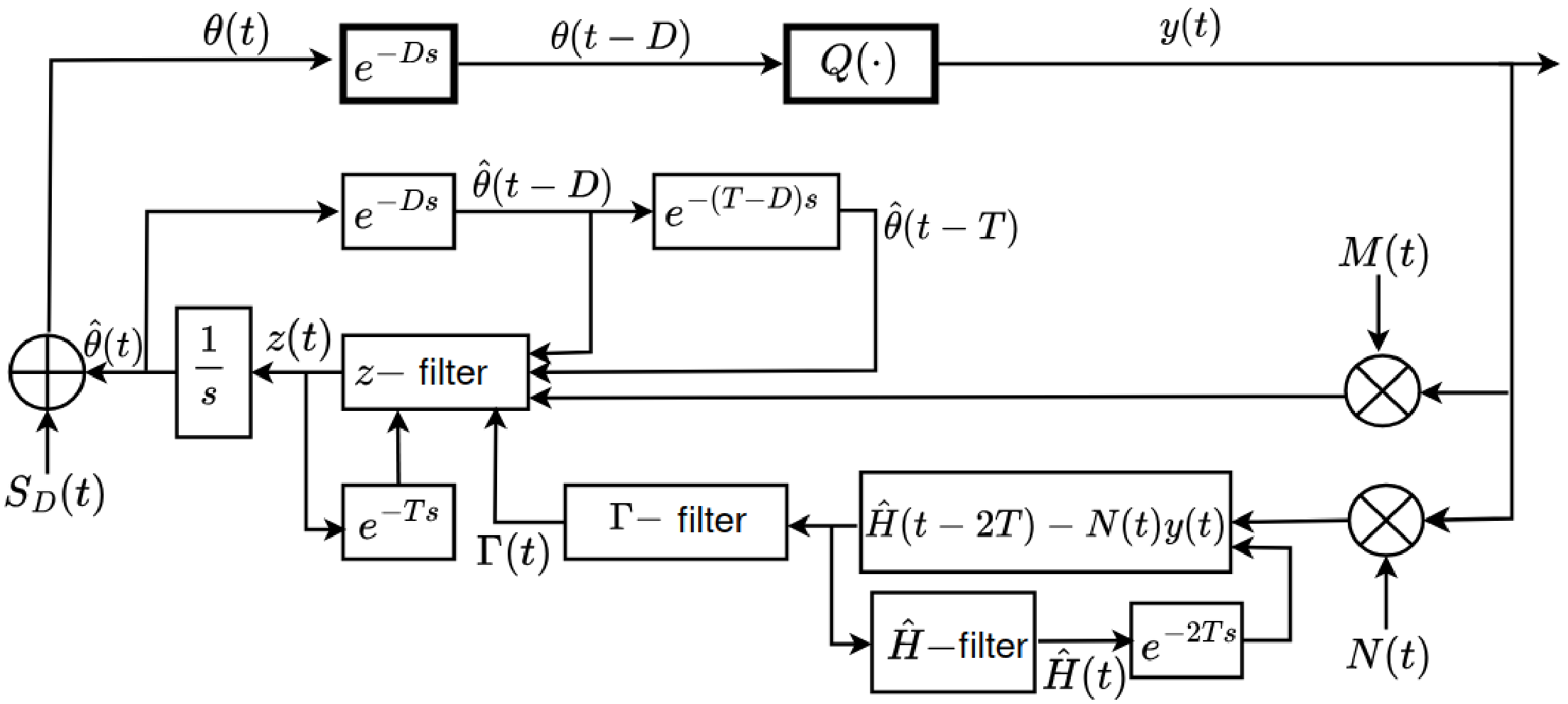}
  \caption{Scheme of the KHV PT-ES  \eqref{eq:Newton_ES_static_maps1_prescribed-pointwise}-\eqref{eq:Newton_ES_static_maps3_prescribed-pointwise_Hessian} with  map delay $D$ and feedback delays $T$ and $2T$.}
  \label{PT_ES_diagram}
\end{figure*}
\begin{figure}[t!]
\centering
  \includegraphics[width=0.78\columnwidth]{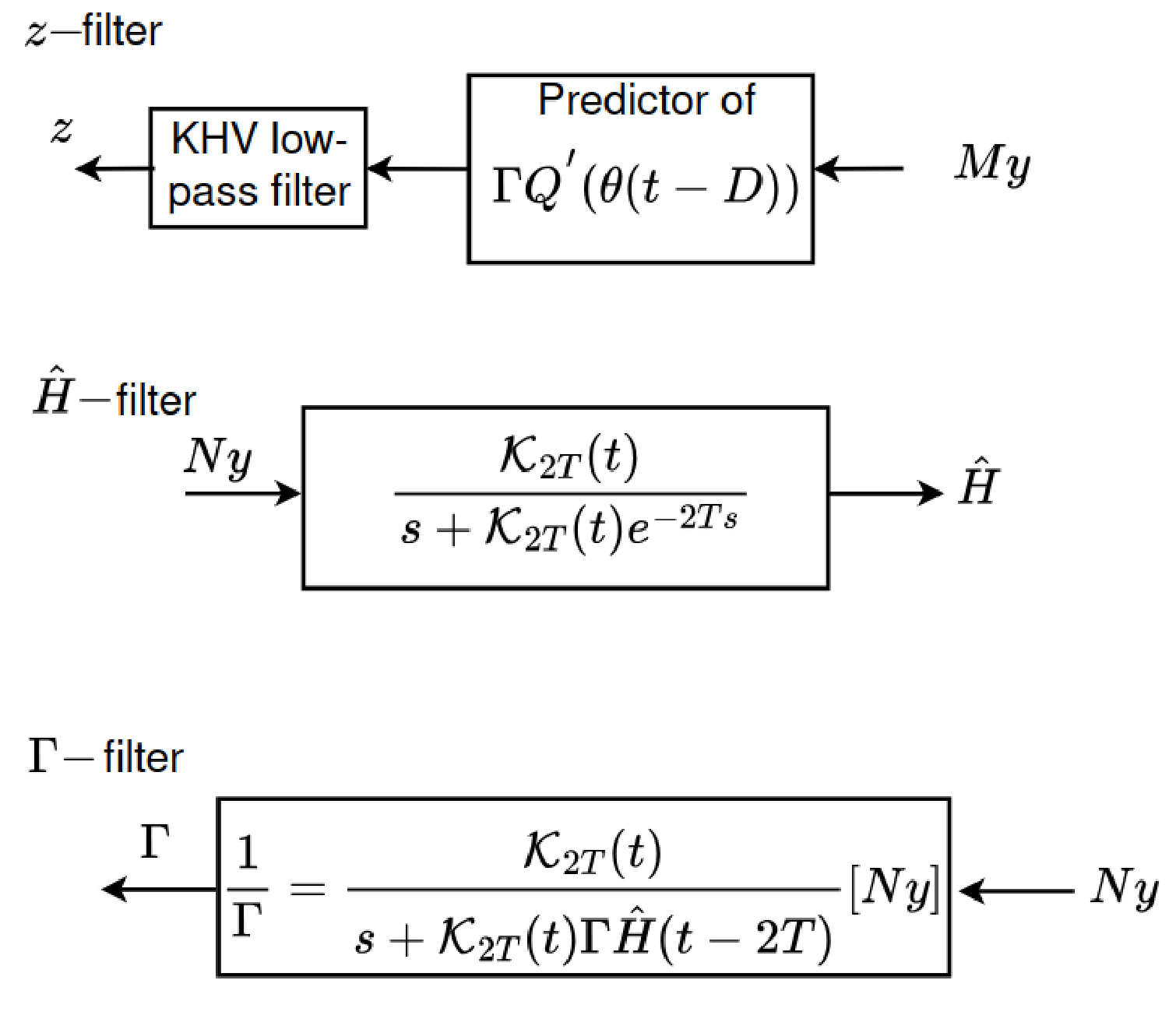}
  \caption{
 Details of the three ``filters'' in Figure \ref{PT_ES_diagram}. Notation is abused copiously (simultaneous time-domain and Laplace-domain nomenclature) to convey intuition.  
\ \ul{Top:} The $z$-filter, given  in \eqref{eq:Newton_ES_static_maps2_prescribed-pointwise}, mirrors the structure of the exponential filter \eqref{eq:Newton_ES_static_maps2V2}. The $z$-filter is fed by $M(t)y(t)$,  an estimate of the ``scalar gradient'' $Q'(\theta(t-D))$. This estimate  drives a predictor for the approximation of a ``Newton update'' $\Gamma Q'(\theta(t-D))$, cascaded into a KHV-type low-pass filter (prescribed time, employing time-periodic gains and delay feedback). The approximation  $\Gamma(t) M(t) y(t) \approx \Gamma(t) Q'(\theta(t-D))$ is the product of $\Gamma$, an estimate of the inverse Hessian $1/H^{\star}$, and $My$, an approximation of the scalar gradient $Q'(\hat\theta(t-D)) = H^{\star} \tilde\theta(t-D)$, as in \eqref{eq:Newton_ES_static_maps2V2}. The KHV filter plays a role analogous to the exponential filter $c/(s+c)$ in \eqref{eq:Newton_ES_static_maps2_prescribed-pointwise}. Although not needed in practice--neither in \eqref{eq:Newton_ES_static_maps2V2} nor in \eqref{eq:Newton_ES_static_maps2_prescribed-pointwise}--this low-pass filtering enables the use of the averaging theorem in \cite{Hale1990}, as the closed-loop system without it does not reduce to a standard RFDE. 
\ \ul{Middle:} The $\hat H$-filter, given in \eqref{eq:Newton_ES_static_maps3_prescribed-pointwise_Hessian}, is a low-pass filter of the KHV kind applied to the Hessian estimate $N(t)y(t)$.  
\ \ul{Bottom:} The $\Gamma$-filter in \eqref{eq:Newton_ES_static_maps3_prescribed-pointwise_Riccati} is a KHV version of the exponential Riccati-type filter in \eqref{eq:Newton_ES_static_maps3V2}, estimating in prescribed time the inverse Hessian $1/H^{\star}$. Since $\hat H$ approximates $H^{\star}$ and $\Gamma$ approximates $1/H^{\star}$, the product $\Gamma(t) \hat H(t - 2T)$  should be regarded as $\approx 1$, in the sense of design inspiration.}
  \label{PT_ES_diagram_filters}
\end{figure}

\begin{align}
\dot{\hat{\theta}}(t)=&  z(t),  \label{eq:Newton_ES_static_maps1_prescribed-pointwise} \\
\dot{z}(t) =&-2\mathcal{K}_{T}(t)z(t-T)
\nonumber\\ &
+\mathcal{L}_{T}(t) \left[ \hat{\theta}(t-D) - \hat{\theta}(t-T)
-\Gamma(t)M(t)y(t)  \right]
,  \label{eq:Newton_ES_static_maps2_prescribed-pointwise} \\
\dot{\Gamma}(t)=&\mathcal{K}_{ 2T}(t) \Gamma^2(t)\left[\hat{H}(t- 2T)  -  N(t)y(t)\right],  
\label{eq:Newton_ES_static_maps3_prescribed-pointwise_Riccati} \\
\dot{\hat{H }}(t)=&-\mathcal{K}_{ 2T}(t)\left[\hat{H}(t-  2T)  -N(t)y(t)\right], \label{eq:Newton_ES_static_maps3_prescribed-pointwise_Hessian}
\end{align} 
where the static map $y(t)$ is given in \eqref{eq:outout_function_static_maps},  $\hat{\theta}(t)\in\mathbb{R}$ is the learning dynamics, $z(t)\in\mathbb{R}$ is the state of a filter-like aiming at estimating the gradient,  the time-varying gains $\mathcal{K}_T(t)$ and $\mathcal{L}_{T}(t)$, are defined, respectively  in \eqref{eq:mathcalK_T} and \eqref{eq:mathcalL_T}, and are generated by the oscillator \eqref{eq:oscillator_system_periodic_time-varying_feedbacks}. 
 The dither signals $S_D(t)$, $M(t)$, $N(t)$ are defined, respectively, in  \eqref{eq:dither_signal_as_tiago_S}, \eqref{eq:dither_signal_as_tiago_M} and \eqref{eq:dither_signal_as_tiago_N}.
 
In the above dynamics,  $\Gamma(t)\in\mathbb{R}$ is the state of a ``Riccati-like" filter  aiming at estimating the inverse of the Hessian of the cost function. It involves, in addition,  the past values of $\hat{H}$ (which estimates the Hessian).  Recall that $\Gamma = \hat{H}^{-1}$ and that \eqref{eq:Newton_ES_static_maps3_prescribed-pointwise_Riccati} is obtained using  $\dot{\Gamma}= - \Gamma^2 \dot{\hat{H}}$.
  The time-periodic gain $\mathcal{K}_{ 2T}(t)$ is defined in \eqref{eq:mathcalK_tau} and is generated by the oscillator \eqref{eq:oscillator_system_periodic_time-varying_feedbacks_tau}.
 We consider  initial conditions  $\phi:[- 2T,0]  \rightarrow \mathbb{R}^4$, for  \eqref{eq:oscillator_system_periodic_time-varying_feedbacks}, where  $(\hat{\theta}^0,z^0,\Gamma^0,\hat{H}^0)^{\top}= (\hat{\theta}^0,z^0,(\hat{H}^0)^{-1},\hat{H}^0)^{\top}=\phi(s)$, with $\phi \in C^{0}([- 2T,0]; \mathbb{R}^{4})  $ for \eqref{eq:Newton_ES_static_maps1_prescribed-pointwise}-\eqref{eq:Newton_ES_static_maps3_prescribed-pointwise_Hessian}, and $\hat{H}^0 \neq 0$ for all $s\in [- 2T,0] $.  In addition,     $\zeta(0)=(\zeta_1(0),\zeta_2(0))^{\top} \in \mathbb{S}^1$, and $\xi(0)=(\xi_1(0),\xi_2(0))^{\top} \in \mathbb{S}^1$ for \eqref{eq:oscillator_system_periodic_time-varying_feedbacks_tau}. 
  Figure \ref{PT_ES_diagram} depicts the scheme of the KHV PT-ES \eqref{eq:Newton_ES_static_maps1_prescribed-pointwise}-\eqref{eq:Newton_ES_static_maps3_prescribed-pointwise_Hessian} with  map delay $D$ and feedback delays $T$ and $2T$. Figure \ref{PT_ES_diagram_filters} shows the details of the three ``filters"  in  Figure \ref{PT_ES_diagram}.
    
 The different delays and ``time-periods" involved in the PT extremum seeking \eqref{eq:Newton_ES_static_maps1_prescribed-pointwise}-\eqref{eq:Newton_ES_static_maps3_prescribed-pointwise_Hessian} are such that the following ordering holds: 
$\frac{2\pi}{\omega} \ll D <T,$
 if  $D>0$, and $\frac{2\pi}{\omega} \ll T$  if $D=0$.
In the sequel, we set $T^{\star}>0$ as the prescribed time of the KHV PT-ES  algorithm \eqref{eq:Newton_ES_static_maps1_prescribed-pointwise}-\eqref{eq:Newton_ES_static_maps3_prescribed-pointwise_Hessian} and we characterize the feedback delay  $T$, in terms to the  map delay $D$, $\delta_{\zeta}$, $\delta_{\xi}$ (defined in \eqref{eq:definition_of_delta-phase_of_oscillators} and \eqref{eq:definition_of_delta-xi-phase_of_oscillators}, respectively), and the prescribed-time $T^{\star}$. Specifically,  for a given $D\geq 0$, $\delta_{\zeta} \in (-1,1]$, $\delta_{\xi} \in (-1,1]$, we pick any prescribed number  $T^{\star}>0$ such that
 \begin{equation}\label{eq:restriction_choice_T_star}
 T^{\star}>\left(2 - \delta_{\zeta}  + 2\Big \lceil \frac{4-(2\delta_{\xi} - \delta_{\zeta}) }{2}\Big\rceil\right)D.
 \end{equation}
  Then, we select
\begin{equation}\label{eq:characteriation_T}
T=
\frac{2D+T^{\star}}{4 - \delta _{\zeta} + 2\Big \lceil \frac{4-(2\delta_{\xi} - \delta_{\zeta}) }{2}\Big\rceil}.
\end{equation}
 Notice that by replacing \eqref{eq:restriction_choice_T_star} in \eqref{eq:characteriation_T} yields $T>D$.  Suppose the  limiting scenario in which $T\leq D$. Consider for example  $T = \frac{D}{(4-2\delta_{\xi})}$. This implies   that the  Hessian and its inverse need to be estimated within at least   $4T-2\delta_{\xi} T=D$ units of time, using \eqref{eq:Newton_ES_static_maps3_prescribed-pointwise_Riccati}-\eqref{eq:Newton_ES_static_maps3_prescribed-pointwise_Hessian} (as we will prove later based on the result in Lemma \ref{eq:Prt_Target_system}). 
 Therefore, it is necessary for  $T$  to be large enough--- at least longer than the map delay, i.e., $T>D$.  This can be achieved by prescribing $T^{\star}$ according to \eqref{eq:restriction_choice_T_star}.

\subsubsection{Error dynamics and averaging  }
Consider the following change of variables $\tilde{\theta}=\hat{\theta}- \theta^\star$, $\tilde{\Gamma}=\Gamma - H^{\star -1}$, $\tilde{H}=\hat{H}-H^\star$. From \eqref{eq:Newton_ES_static_maps1_prescribed-pointwise}-\eqref{eq:Newton_ES_static_maps3_prescribed-pointwise_Hessian}, the error dynamics reads as follows:
\begin{align}
\dot{\tilde{\theta}}(t)=&  z(t),  \label{eq:Newton_ES_static_maps1_prescribed-pointwise_error} \\
\dot{z}(t) =&- 2\mathcal{K}_{T}(t)z(t-T)    
\nonumber  \\ & 
+\mathcal{L}_{T}(t) 
\left[\tilde{\theta}(t-D)
- \tilde{\theta}(t-T) 
\right.\nonumber\\ & \left. 
- (\tilde{\Gamma}(t)+H^{\star-1})M(t)y(t) 
\right], 
\label{eq:Newton_ES_static_maps2_prescribed-pointwise_error} 
\\
\dot{\tilde{\Gamma}}(t)=&
\mathcal{K}_{ 2T}(t)
(\tilde{\Gamma}(t)+H^{\star-1})^2 
\nonumber \\ & 
\times\left[\tilde{H}(t- 2T)+H^\star
 - N(t)y(t)\right], 
\label{eq:Newton_ES_static_maps3_prescribed-pointwise_Riccati_error} \\
\dot{\tilde{H }}(t)=&-\mathcal{K}_{ 2T}(t)\left[\tilde{H}(t- 2T)+H^\star
 - N(t)y(t)\right], \label{eq:Newton_ES_static_maps3_prescribed-pointwise_Hessian_error}
\end{align} 
with initial conditions $(\tilde{\theta}^0,z^0,\tilde{\Gamma}^0,\tilde{H}^0)^{\top}= \tilde{\phi}(s)$, $ \tilde{\phi} \in C^{0}([- 2T,0]; \mathbb{R}^{4})$, with $\tilde{\theta}^0=\hat{\theta}^0-\theta^{\star}$,  $\tilde{\Gamma}^0=(\hat{H}^0)^{-1} - H^{\star -1}$, and $\tilde{H}^0=\hat{H}^0-H^{\star}$.  Notice that $\tilde{\Gamma}^0 + H^{\star-1} = (\tilde{H}^0 + H^{\star})^{-1}$  for all $s\in [-2T,0]$.   As we have mentioned in  Remark \ref{remk:state-dependent-signals-time-invariantES}, the signals $\mathcal{K}_T(t)=\mathcal{K}_T(\zeta_1(t))$, $\mathcal{L}_T(t)=\mathcal{L}_T(\zeta_1(t),\zeta_2(t))$ and $\mathcal{K}_{ 2T}(t)=\mathcal{K}_{ 2T}(\xi_1(t))$   are seen as state dependent signals generated by the oscillators.
The delays appearing in the above system can be  ordered  as $D < T < 2T$.
For \eqref{eq:Newton_ES_static_maps1_prescribed-pointwise_error}-\eqref{eq:Newton_ES_static_maps3_prescribed-pointwise_Hessian_error}, consider the following  RFDE with fast oscillations:  
\begin{equation}\label{eq:Compact_form_error_system}
\dot{X}(t)= f(\omega t,X_t),
\end{equation} 
with initial condition   $X(s)=\varphi(s) = (\tilde{\phi},\zeta(0),\xi(0))^{\top} \in C^{0}([- 2T,0];\mathbb{R}^4)\times \mathbb{S}^{1}\times \mathbb{S}^{1}$, and
  state variable    $X=(\tilde{\theta},z,\tilde{\Gamma},\tilde{H}, \zeta_1,\zeta_2,\xi_1,\xi_2)^{\top} \in \mathbb{R}^{8}$, $X_t \in C^{0}([- 2T,0];\mathbb{R}^{4})\times \mathbb{S}^{1}\times \mathbb{S}^{1}$, $X_t(s)=X(t+s)$ for all $s \in [- 2T,0]$. 
 Therefore, we have $f=(f_1,f_2,f_3,f_4,f_5,f_6,f_7,f_8)^{\top}$  with
\begin{align}
f_1(\omega t, X_t)&=X_2(t),  \label{eq:RFDE-f1}\\
f_2(\omega t,X_t)&= - 2\mathcal{K}_{T}(X_5(t))X_2(t-T) \nonumber \\
&   - \mathcal{L}_{T}(X_5(t),X_6(t))X_1(t-T)\nonumber \\ 
& +\mathcal{L}_{T}(X_5(t),X_6(t)) X_1(t-D) \nonumber \\  
 &  -\mathcal{L}_{T}(X_5(t),X_6(t))  [X_3(t)+H^{\star-1}] \nonumber \\
  &\times [\tfrac{2}{a}\sin(\omega t)][ y^\star + \tfrac{H^\star}{2}\left[X_1(t-D) +a\sin(\omega t) \right]^{2}], \label{eq:RFDE-f2} \\
f_3(\omega t,X_t)&=\mathcal{K}_{ 2T}(X_7(t)) [X_4(t- 2T)+H^\star][X_3(t)+H^{\star-1}]^2 \nonumber \\
&   - \mathcal{K}_{ 2T}(X_7(t))[X_3(t)+H^{\star-1}]^2[\tfrac{16}{a^2}\left[\sin^2(\omega t) - \tfrac{1}{2} \right]] \nonumber\\
& \times [ y^\star + \tfrac{H^\star}{2}\left[X_1(t-D) +a\sin(\omega t) \right]^{2}],  \label{eq:RFDE-f3}\\
f_4(\omega t,X_t)&=-\mathcal{K}_{ 2T}(X_7(t))[X_4(t-  2T) + H^\star] \nonumber \\ 
 &+ \mathcal{K}_{ 2T}(X_7(t))[\tfrac{16}{a^2}\left[\sin^2(\omega t) - \tfrac{1}{2} \right]] \nonumber\\
&\times [ y^\star + \tfrac{H^\star}{2}\left[X_1(t-D) +a\sin(\omega t) \right]^{2}], \label{eq:RFDE-f4}\\ 
f_5(\omega t,X_t)&=\tfrac{\pi}{T} X_6(t), \label{eq:RFDE-f5}\\ 
f_6(\omega t,X_t)&=-\tfrac{\pi}{T} X_5(t), \label{eq:RFDE-f6}\\
f_7(\omega t,X_t)&=\tfrac{\pi}{ 2T} X_8(t),\label{eq:RFDE-f7} \\
f_8(\omega t,X_t)&=-\tfrac{\pi}{ 2T} X_7(t). \label{eq:RFDE-f8}
\end{align}
 The compact representation $\dot{X}(t) = f(\omega t, X_t)$ will be useful for applying the averaging theorem in infinite-dimensional settings \cite{Hale1990}, as will be revisited in Section~\ref{proof_of_main_result}. The prescribed-time convergence properties
of the closed-loop system \eqref{eq:Newton_ES_static_maps1_prescribed-pointwise_error}-\eqref{eq:Newton_ES_static_maps3_prescribed-pointwise_Hessian_error} can be investigated through the corresponding average system that we deduce by  using the properties \eqref{property:averaging_get_gradient}-\eqref{property:averaging_get_Hessian} in conjunction with the averaging operation \cite{Hale1990}. Hence,   
the average of the system $\dot{X}(t)=f(\omega t,X_t)$,   with $f$ stated above  over the period $\Pi=\tfrac{2\pi}{\omega}$, results in a closed-loop system of the form  
\begin{equation}\label{eq:Compact_form_Averaged_error_system}
\dot{X}^{\rm av}(t)=F^{\rm av}(X^{\rm av}_{t}),
\end{equation}
 with initial condition   $X^{\rm av}(s)=\varphi(s) = (\tilde{\phi},\zeta(0),\xi(0))^{\top} \in C^{0}([- 2T,0];\mathbb{R}^4)\times  \mathbb{S}^{1}\times \mathbb{S}^{1}$,  state variable    $X^{\rm av}=(\tilde{\theta}_{\rm av},z_{\rm av},\tilde{\Gamma}_{\rm av},\tilde{H}_{\rm av}, \zeta_{1},\zeta_{2}, \xi_{1},\xi_{2})^{\top}$,  $X_t^{\rm av}(s)=X^{\rm av}(t+s)$ for all $s \in [- 2T,0]$ and    $ F^{\rm av}(\psi)=\tfrac{1}{\Pi}\int_{0}^{\Pi}f(s,\psi)ds $.  
In particular, we display    
\begin{align}
\dot{\tilde{\theta}}_{\rm av}(t)=&  z_{\rm av}(t),  \label{eq:Newton_ES_static_maps1_prescribed-pointwise_error_averaged} \\
\dot{z}_{\rm av}(t) =& - 2\mathcal{K}_{T}(t)z_{\rm av}(t-T)\nonumber  \\
 &   - \mathcal{L}_{T}(t)
 \left[\tilde{\theta}_{\rm av}(t-T)    
 +
 \tilde{\Gamma}_{\rm av}(t)H^{\star}\tilde{\theta}_{\rm av}(t-D)\right], 
  \label{eq:Newton_ES_static_maps2_prescribed-pointwise_error_averaged} \\
\dot{\tilde{\Gamma}}_{\rm av}(t)=&\mathcal{K}_{ 2T}(t)\tilde{H}_{\rm av}(t- 2T)(\tilde{\Gamma}_{\rm av}(t)+H^{\star-1})^2, \label{eq:Newton_ES_static_maps3_prescribed-pointwise_Riccati_error_averaged}  \\
\dot{\tilde{H }}_{\rm av}(t)=&-\mathcal{K}_{ 2T}(t)\tilde{H}_{\rm av}(t-  2T).   \label{eq:Newton_ES_static_maps3_prescribed-pointwise_Hessian_error_averaged}
\end{align} 
with initial conditions $(\tilde{\theta}_{\rm av}^0,z_{\rm av}^0,\tilde{\Gamma}_{\rm av}^0,\tilde{H}_{\rm av}^0)^{\top}=(\tilde{\theta}^0,z^0,\tilde{\Gamma}^0,\tilde{H}^0)^{\top}  \in C^{0}([- 2T,0]; \mathbb{R}^{4})$.   Notice that the solutions to \eqref{eq:Newton_ES_static_maps3_prescribed-pointwise_Riccati_error_averaged} and \eqref{eq:Newton_ES_static_maps3_prescribed-pointwise_Hessian_error_averaged} satisfy the following relation: 
\begin{equation}\label{eq:relation_Gamma_H}
\tilde{\Gamma}_{\rm av}(t) + H^{\star-1} = \frac{1}{\tilde{H}_{\rm av}(t) + H^{\star}},
\end{equation}
which implies $\tilde{\Gamma}_{\rm av}(t) = \tfrac{-H^{\star-1} \tilde{H}_{\rm av}(t)}{\tilde{H}_{\rm av}(t)+H^{\star}}$, for all $t\geq 0$.  This holds true since the initial conditions are such that $\tilde{\Gamma}^0_{\rm av}(s) + H^{\star-1} = (\tilde{H}^0_{\rm av}(s) + H^{\star})^{-1}$  for all $s\in [-2T,0]$. 

\subsubsection{Prescribed-time  convergence  of the averaged error dynamics}
Since the dynamics of the oscillators $\zeta$, $\xi$ given in \eqref{eq:oscillator_system_periodic_time-varying_feedbacks} and \eqref{eq:oscillator_system_periodic_time-varying_feedbacks_tau}, respectively,  render forward invariant  the set $\mathbb{S}^{1}$ (see Remark \ref{remk:forward_invariance_oscillators}), we focus mainly on the  convergence  properties of the states $(\tilde{\theta}_{\rm av},z_{\rm av},\tilde{\Gamma}_{\rm av},\tilde{H}_{\rm av})^{\top}$. 
 Before we present the prescribed-time convergence result for \eqref{eq:Newton_ES_static_maps1_prescribed-pointwise_error_averaged}-\eqref{eq:Newton_ES_static_maps3_prescribed-pointwise_Hessian_error_averaged}, we need first a key result on the prescribed-time input-to-state stabilization of single integrator systems by time-periodic delayed feedback.

\begin{lemma}\label{Lemma_PrT_Pointwise_delay}
 Consider the   following  one-dimensional  system 
\begin{equation}\label{eq:single_integrator}
\dot{x}(t)=u(t) + v(t),
\end{equation}
where $u(t) \in \mathbb{R}$ is the control input and    $v \in \mathcal{C}^{0}(\mathbb{R}_+;\mathbb{R})$ is an arbitrary input. Let
 $T>0$ be a prescribed number. The solution of the  closed-loop system \eqref{eq:single_integrator} with  the following \textit{$2 T$-periodic} time-varying \textit{pointwise} delayed feedback
\begin{equation}\label{eq:time-varying_delayed-feedback_key_result}
u(t)= -\mathcal{K}_{T}\left(t\right)x(t-T),
\end{equation}
where $\mathcal{K}_{T}(t)$ is defined  by \eqref{eq:mathcalK_T},   with arbitrary initial condition  $x(s)=\phi(s)$, $\phi \in C^0([-T,0];\mathbb{R})$  satisfies 
 \begin{equation}\label{eq:bound_PrT_single_integrator}
\begin{split}
 \vert x(t) \vert \leq &\exp\left(6 \right)\sigma\left(t-(2T - \delta_{\zeta} T)  \right)\max_{-T \leq s\leq 0}(\vert \phi(s) \vert)  \\
 & \hskip 1cm + 6T \exp\left(6 \right) \sup_{\max\{0,t-(2T-\delta_{\zeta}  T) \}\leq s\leq t} (\vert v(s) \vert), 
\end{split} 
 \end{equation}
\begin{equation}\label{eq:bound_PrT_single_integrator-control}
\begin{split}
 \vert u(t) \vert \leq & \tfrac{2}{T}\exp\left(6 \right)\sigma\left(t-(2T - \delta_{\zeta} T)  \right)\max_{-T \leq s\leq 0}(\vert \phi(s) \vert) \\
 & \hskip 1cm + 12 \exp\left(6 \right) \sup_{\max\{0,t-(2T-\delta_{\zeta}  T) \}\leq s\leq t} (\vert v(s) \vert ),
\end{split}
\end{equation}
where  
\begin{equation}\label{eq:heaviside_function}
    \sigma(s) = \begin{cases}1& \mbox{if} \quad s < 0\\ 0
    & \mbox{if} \quad s \geq 0.   \end{cases}
\end{equation}
 and $\delta_{\zeta} \in (-1,1]$ is given by \eqref{eq:definition_of_delta-phase_of_oscillators}. 
Moreover,   if the  input  $v(t)$ completely vanishes at $t^{'} > 0$,  i.e.,  $v(t) =0$ for  $t \geq t^{'} > 0$,  then, the following estimate holds:
\begin{equation}\label{eq:bound_PrT_single_integrator-control-vanishing_dist}
\begin{split}
 \vert x(t) \vert \leq &\exp\left(6 \right)\sigma\Big(t-\left(2T-\delta_{\zeta} T + 2T\Big\lceil \tfrac{t^{'}+\delta_{\zeta} T}{2T}\Big\rceil\right)  \Big)\\ & \times \Big[ \max_{-T \leq s\leq 0}(\vert \phi(s) \vert) + 6T \sup_{0\leq s\leq \min\{t,t^{'}\}} (\vert v(s)  \vert)\Big], 
\end{split} 
\end{equation}
which yields  $x(t)=0$ for all $t \geq   2 T - \delta_{\zeta} T+ 2T\Big \lceil \frac{t^{'}+\delta_{\zeta}  T}{2T}\Big\rceil$. 
\end{lemma}
\begin{proof}
See Appendix \ref{proof_of_lemma1}.
\end{proof}
It is important to understand the effect of the size of $T$ in the gain \eqref{eq:time-varying_delayed-feedback_key_result}, and also in the settling time $2T - \delta_{\zeta} T$, on the performance of the closed-loop system. The ISS gain from the input $v$ to the state $x$ is proportional to $T$ and hence, as $T$ is decreased, and the gain in \eqref{eq:time-varying_delayed-feedback_key_result} consequently increases, the gain from $v$ to $x$ improves. However, there is an adverse effect of increasing the gain $1/T$ on the control magnitude. While the ISS gain from $v$ to $u$ is independent of $T$, the finite-time transient portion of the bound on $u(t)$ is proportional to $1/T$, meaning that, for improved performance of $x(t)$, including shorter setting time $2T-\delta_{\zeta} T$, price is paid in the magnitude of $u(t)$, at least initially. 

\begin{remark}\label{remark_on_historical_issues}
The one-dimensional control system \eqref{eq:single_integrator} with time-varying pointwise delayed feedback \eqref{eq:time-varying_delayed-feedback_key_result}  is  a special case of the one-dimensional control  system  with time-varying \textit{distributed}  delayed feedback studied in \cite[Lemma 2.9]{Karafyllis2006}. Therefore, applying \cite[Lemma 2.9]{Karafyllis2006} along  with the seminal result in \cite[Chapter 3, pp 87-88]{Hale1993}, the conclusion would follow: we could  obtain a similar bound to \eqref{eq:bound_PrT_single_integrator}, and the ISS prescribed-time stability result.  
Nevertheless, for the sake of completeness and conciseness, we adapt the ideas of the proof of  \cite[Lemma 2.9]{Karafyllis2006} to our specific case of time-varying pointwise delay feedback, with gain $\mathcal{K}_{T}$ defined as in \eqref{eq:mathcalK_T},  and including the effect of a vanishing input. The proof is given in Appendix \ref{proof_of_lemma1}.
\end{remark}

\begin{lemma}\label{eq:Prt_Target_system}
 Let $D\geq 0$ be given. Let us choose   the  prescribed time $T^{\star}$ according to condition \eqref{eq:restriction_choice_T_star}.     Consider the closed-loop system \eqref{eq:Newton_ES_static_maps1_prescribed-pointwise_error_averaged}-\eqref{eq:Newton_ES_static_maps3_prescribed-pointwise_Hessian_error_averaged},   with  gains $\mathcal{K}_T(t)$,  $\mathcal{L}_{T}(t)$ and $\mathcal{K}_{ 2T}(t)$  defined   in \eqref{eq:mathcalK_T}, \eqref{eq:mathcalL_T} and \eqref{eq:mathcalK_tau}, respectively,   which are generated by the oscillators \eqref{eq:oscillator_system_periodic_time-varying_feedbacks} and \eqref{eq:oscillator_system_periodic_time-varying_feedbacks_tau}, with arbitrary initial conditions in $\mathbb{S}^{1}$ yielding   $\delta_{\zeta}$ and $\delta_{\xi}$  as in  \eqref{eq:definition_of_delta-phase_of_oscillators} and \eqref{eq:definition_of_delta-xi-phase_of_oscillators}, respectively. Let $T$ be selected according to \eqref{eq:characteriation_T}. The  solutions of the closed-loop system \eqref{eq:Newton_ES_static_maps1_prescribed-pointwise_error_averaged}-\eqref{eq:Newton_ES_static_maps3_prescribed-pointwise_Hessian_error_averaged}  are such that $\tilde{\Gamma}_{\rm av}(t)=0$ and $\tilde{H}_{\rm av}(t)=0$ for all $t \geq  4T - \delta_{\xi} 2T$; and $\tilde{\theta}_{\rm av}(t)=0$ and $z_{\rm av}(t)=0$   for all $t\geq T^{\star}+2D$. Moreover, there exists a class-$\mathcal{K}_{\infty}$ function $\gamma$ such that the  solutions of the closed-loop system \eqref{eq:Newton_ES_static_maps1_prescribed-pointwise_error_averaged}-\eqref{eq:Newton_ES_static_maps3_prescribed-pointwise_Hessian_error_averaged} with initial function $(\tilde{\theta}_{\rm av}^0,z_{\rm av}^0,\tilde{\Gamma}_{\rm av}^0,\tilde{H}_{\rm av}^0)^{\top}=\tilde{\phi} \in C^{0}([- 2T,0]; \mathbb{R}^{4})$ with $\tilde{\Gamma}^0_{\rm av} + H^{\star-1} = (\tilde{H}^0_{\rm av} + H^{\star})^{-1}$ and $\tilde{H}^0_{\rm av} + H^{\star} \neq 0$ for all $s \in [-2T,0]$, satisfy  
\begin{equation}\label{eq:bound_Kinfty_lemma}
\vert (\tilde{\theta}_{\rm av},z_{\rm av},\tilde{\Gamma}_{\rm av},\tilde{H}_{\rm av})^{\top} \vert \leq \gamma\left(\Vert \tilde{\phi} \Vert\right), \quad \forall t \geq 0.
\end{equation} 
\end{lemma}
\begin{IEEEproof}
Consider the following  change of variables, inspired by the backstepping procedure, and specialized here for the double integrator system\footnote{
We recall that for higher-order systems with time-varying delayed feedbacks, the backstepping procedure can also be applied, see e.g.,   \cite[Chapter 6.7]{Iasson_book2011}  and  \cite{Ding_stronPT-nonlinear2024}.
}  \eqref{eq:Newton_ES_static_maps1_prescribed-pointwise_error_averaged}-\eqref{eq:Newton_ES_static_maps2_prescribed-pointwise_error_averaged}: 
\begin{align}
x_1(t)=&\tilde{\theta}_{\rm av}(t),\label{change_of_variable_targetsystem1} \\
x_2(t)=&\mathcal{K}_{T}(t)\tilde{\theta}_{\rm av}(t-T)+ z_{\rm av}(t) \label{change_of_variable_targetsystem2}.
\end{align}
 Using this change of variable together with   \eqref{eq:mathcalL_T} and  the fact $\mathcal{K}_{T}(t)\mathcal{K}_{T}(t-T) = 0$ (see \ref{eq:property_K_T_and delayed_K}), we obtain that   
 system \eqref{eq:Newton_ES_static_maps1_prescribed-pointwise_error_averaged}-\eqref{eq:Newton_ES_static_maps3_prescribed-pointwise_Hessian_error_averaged} is transformed into the following \textit{target} system:
\begin{align}
\dot{x}_{1}(t)=& - \mathcal{K}_{T}(t)x_1(t-T) + x_2(t),  \label{eq:Newton_ES_static_maps1_prescribed-pointwise_error_averaged_target} \\
\dot{x}_{2}(t) =&  - \mathcal{K}_{T}(t)x_2(t-T)   -\mathcal{L}_{T}(t)\tilde{\Gamma}_{\rm av}(t)H^{\star}x_{1}(t-D),  \label{eq:Newton_ES_static_maps2_prescribed-pointwise_error_averaged_target}  \\
\dot{\tilde{\Gamma}}_{\rm av}(t)=&\mathcal{K}_{ 2T}(t)\tilde{H}_{\rm av}(t- 2T)(\tilde{\Gamma}_{\rm av}(t)+H^{\star-1})^2, \label{eq:Newton_ES_static_maps3_prescribed-pointwise_Riccati_error_averaged_target}  \\
\dot{\tilde{H }}_{\rm av}(t)=&-\mathcal{K}_{ 2T}(t)\tilde{H}_{\rm av}(t-  2T),   \label{eq:Newton_ES_static_maps3_prescribed-pointwise_Hessian_error_averaged_target}
\end{align}
with initial conditions  $(x_1^0,x_2^0,\tilde{\Gamma}_{\rm av}^0,\tilde{H}_{\rm av}^0)^{\top}  \in C^{0}([- 2T,0]; \mathbb{R}^{4})  $.
 By Lemma \ref{Lemma_PrT_Pointwise_delay}, we obtain the following estimate for the solution to \eqref{eq:Newton_ES_static_maps3_prescribed-pointwise_Hessian_error_averaged_target}: 
\begin{equation}\label{eq:PrT_Hessian_estimator_average_in_proof}
\vert \tilde{H}_{\rm av}(t)  \vert  \leq \exp\left(6 \right)\sigma\left(t- (4T-\delta_{\xi} 2T) \right)\max_{- 2T \leq s\leq 0}(\vert \tilde{H}_{\rm av}^0(s) \vert),
\end{equation}
for $t\geq 0$ where  $\tilde{H}_{\rm av}(0)=\tilde{H}_{\rm av}^0(s)$, $s \in [- 2T,0]$ and
 $\sigma(\cdot)$ is defined in \eqref{eq:heaviside_function}. Therefore, $\tilde{H}_{\rm av}(t) = 0 $ for $t \geq 4T - \delta_{\xi} 2T$.   Consequently, using \eqref{eq:PrT_Hessian_estimator_average_in_proof}  and \eqref{eq:relation_Gamma_H}, we also obtain $\tilde{\Gamma}_{\rm av}(t) = 0 $ for $t \geq  4T  -\delta_{\xi} 2T$. Next, notice that the time-delay system  \eqref{eq:Newton_ES_static_maps2_prescribed-pointwise_error_averaged_target} is cascaded by $\tilde{\Gamma}_{\rm av}(t)$, which is the solution of \eqref{eq:Newton_ES_static_maps3_prescribed-pointwise_Riccati_error_averaged_target}, and that for all $t \geq  4T-\delta_{\xi} 2T$, $\tilde{\Gamma}_{\rm av}(t)=0$.  Meanwhile, the term  $x_{1}(t-D)$ in \eqref{eq:Newton_ES_static_maps2_prescribed-pointwise_error_averaged_target} remains bounded for all $t \in[0, 4T-\delta_{\xi} 2T]$, with $T> D$ and  $\delta_{\xi} \in (-1,1]$. This boundedness is stated in the following claim and proved in  Appendix \ref{proof_claim}.
\begin{claim}\label{claim1}
For the solution of \eqref{eq:Newton_ES_static_maps1_prescribed-pointwise_error_averaged_target}, it holds that
\begin{equation}\label{first_bound_Claim1}
\sup_{-D \leq s \leq 0}(\vert x_1(t+s) \vert) < +\infty, \quad \forall t\in [0, 4T-\delta_{\xi} 2T].
\end{equation} 
\end{claim}
The previous result ensures $x_2 \in \mathcal{C}^{0}(\mathbb{R}_+;\mathbb{R}) $, hence $x_2$ becomes the new  finite-time vanishing input for system \eqref{eq:Newton_ES_static_maps1_prescribed-pointwise_error_averaged_target}. Therefore,  applying  
 repeatedly Lemma \ref{Lemma_PrT_Pointwise_delay} (in particular estimate \eqref{eq:bound_PrT_single_integrator-control-vanishing_dist}), with 
\begin{equation}\label{eq:vanishing_perturbation_v_in_proof}
 v(t)=-\mathcal{L}_{T}(t) \tilde{\Gamma}_{\rm av}(t)H^{\star}x_{1}(t-D),
 \end{equation}
and by  knowing that $v(t)=0$ for $t\geq  4T-\delta_{\xi} 2T$,   we have the following estimates for the solutions to \eqref{eq:Newton_ES_static_maps1_prescribed-pointwise_error_averaged_target} and \eqref{eq:Newton_ES_static_maps2_prescribed-pointwise_error_averaged_target}:
\begin{equation}\label{eq:estimate_x1_afterLemma1}
\begin{split}
\vert x_{1}(t)  \vert  \leq & \sigma \left(t-\tau_1     \right)\Big(\exp\left(6 \right) \max_{- 2T \leq s\leq 0}(\vert  x_{1}^0(s) \vert) \\
& \quad + 6T\exp\left(12 \right) \max_{- 2T \leq s\leq 0}(\vert  x_{2}^0(s) \vert) \\
&  \quad  + 36T^2 \exp\left(12 \right)  \sup_{0\leq s\leq \min\{t,4T-\delta_{\xi} 2T\}} (\vert v(s)  \vert) \Big),
\end{split}
\end{equation}
 and
 \begin{equation}\label{eq:estimate_x2_afterLemma1}
\begin{split}
\vert x_{2}(t)  \vert   \leq &  \sigma \left(t-\tau_2   \right)\Big(\exp\left(6 \right) \max_{- 2T \leq s\leq 0}(\vert  x_{2}^0(s) \vert) \\
& \quad \quad + 6T\exp\left(6 \right)   \sup_{0\leq s\leq \min\{t,4T-\delta_{\xi} 2T\}} (\vert v(s)  \vert) \Big),
\end{split}
\end{equation}
for all $t\geq 0$, where $\tau_1$ and $\tau_2$ are defined as follows:
\begin{equation}\label{eq:tau1_tau2inproof}
\tau_1  = 4T - \delta_\zeta T + 2T \left\lceil \tfrac{4 - (2\delta_\xi - \delta_\zeta)}{2} \right\rceil, \quad
\tau_2  = \tau_1 - 2T.  
\end{equation} 
From \eqref{change_of_variable_targetsystem2}, notice that  $x_2(t)=0$, thus $z_{\rm av}(t)=-\mathcal{K}_{T}(t)\tilde{\theta}_{\rm av}(t-T)$, for $t \geq \tau_2 $. In addition, it holds that $\mathcal{K}_T(t) = 0$ for $t \in [\tau_1, \tau_1  + T]$ and $\tilde{\theta}_{\rm av}(t - T) = 0$ for $t \geq \tau_1  + T$, thus $z_{\rm av}(t) = 0$ for all $t \geq \tau_1 $. Moreover, using  \eqref{change_of_variable_targetsystem1}-\eqref{change_of_variable_targetsystem2},  
 the triangle inequality and the definition of the norm, we  have on the one hand, 
\begin{equation}\label{eq:norm_equivalence_iniitalconditions1}
\Vert (x_1^0,x_2^0)^{\top} \Vert \leq q\left( \Vert (\tilde{\theta}_{\rm av}^0,z_{\rm av}^0)^{\top}   \Vert \right),
\end{equation}
with $q(s)=(2/T +2)s \in\mathcal{K}_{\infty}$. On the other hand,  using again  \eqref{change_of_variable_targetsystem1}-\eqref{change_of_variable_targetsystem2} together with the facts 
\[
\mathcal{K}_T(t) \sigma(t - (\tau_1  + T)) = \mathcal{K}_T(t) \sigma(t - \tau_1),
\] and 
$
\sigma(t - \tau_2 ) \leq \sigma(t - \tau_1) $  for all $t \geq 0 $, we   get,  from \eqref{eq:estimate_x1_afterLemma1}-\eqref{eq:estimate_x2_afterLemma1} and  \eqref{eq:norm_equivalence_iniitalconditions1}, the following estimates:
 \begin{equation}\label{eq:estimate_tilde_theta_av_afterLemma1}
\begin{split}
\vert \tilde{\theta}_{\rm av}(t)  \vert  \leq & \sigma\left(t- \tau_1  \right) \Big((1+6T)\exp\left(12 \right)q\left( \Vert (\tilde{\theta}_{\rm av}^0,z_{\rm av}^0)^{\top}   \Vert \right)  \\
&  \quad \quad   + 36T^2 \exp\left(12 \right)  \sup_{0\leq s\leq \min\{t,4T-\delta_{\xi} 2T\}} (\vert v(s)  \vert) \Big),
\end{split}
\end{equation}
and 
\begin{equation}\label{eq:estimate_z_av_afterLemma1}
\begin{split}
\vert z_{\rm av}(t)  \vert  \leq  &\sigma\left(t- \tau_1 \right) \Big((2/T+13)\exp\left(12 \right)q\left( \Vert (\tilde{\theta}_{\rm av}^0,z_{\rm av}^0)^{\top}   \Vert \right)  \\
&  \quad  + 78T \exp\left(12 \right)  \sup_{0\leq s\leq \min\{t,4T-\delta_{\xi} 2T\}} (\vert v(s)  \vert) \Big),
\end{split}
\end{equation}
for all $t\geq 0$; where $v(t)$ is as in  \eqref{eq:vanishing_perturbation_v_in_proof} with $x_1(t)=\tilde{\theta}_{\rm av}(t)$. Therefore, using \eqref{eq:tau1_tau2inproof} and \eqref{eq:characteriation_T}, we obtain  from \eqref{eq:estimate_tilde_theta_av_afterLemma1} and \eqref{eq:estimate_z_av_afterLemma1} that $\tilde{\theta}_{\rm av}(t) = z_{\rm av}(t) = 0$ for all $t \geq T^\star + 2D$. Finally, combining  \eqref{eq:PrT_Hessian_estimator_average_in_proof}, \eqref{eq:estimate_tilde_theta_av_afterLemma1} and \eqref{eq:estimate_z_av_afterLemma1},  using the boundedness of  \eqref{eq:vanishing_perturbation_v_in_proof}  via \eqref{eq:relation_Gamma_H} and \eqref{first_bound_Claim1}, we have that the estimate \eqref{eq:bound_Kinfty_lemma} follows  with  $\gamma(s)=C_1 q(s) +C_2 s \in \mathcal{K}_{\infty}$ where  $C_1, C_2>0$ are constants depending on $T$, $\vert H^{*} \vert $, and  some positive bound $M$ (from Claim \ref{claim1}). This concludes the proof.  
 \end{IEEEproof}

\section{Main result, discussion and extension to multivariable static maps} \label{proof_of_main_result}

\subsection{Main result}
\begin{theorem}\label{theo:main_ES_algorithm_prescribed-pointwise}
 Let $D\geq 0$ be given.  Let $T^{\star}$ be the prescribed time that meets \eqref{eq:restriction_choice_T_star}.
Consider the KHV PT-ES algorithm \eqref{eq:Newton_ES_static_maps1_prescribed-pointwise}-\eqref{eq:Newton_ES_static_maps3_prescribed-pointwise_Hessian}  with output map \eqref{eq:outout_function_static_maps}  and  time-periodic gains $\mathcal{K}_T(t)$,  $\mathcal{L}_{T}(t)$ and $\mathcal{K}_{ 2T}(t)$  defined   in \eqref{eq:mathcalK_T}, \eqref{eq:mathcalL_T} and \eqref{eq:mathcalK_tau}, respectively,  which are generated by the oscillators \eqref{eq:oscillator_system_periodic_time-varying_feedbacks} and \eqref{eq:oscillator_system_periodic_time-varying_feedbacks_tau}, with arbitrary initial conditions in $\mathbb{S}^{1}$ yielding   $\delta_{\zeta}$ and $\delta_{\xi}$  as in  \eqref{eq:definition_of_delta-phase_of_oscillators} and \eqref{eq:definition_of_delta-xi-phase_of_oscillators}, respectively. Let $T$ be selected according to \eqref{eq:characteriation_T}.    Then,  there exist $\rho >0$  and  a  class-$\mathcal{K}_{\infty}$ function $\gamma$  such that for  any  $\nu \in (0,\rho)$  and any arbitrarily large $\bar{T}\gg T^* + 3D$, there exists $\bar{\omega} >0$  such that for any $\omega > \bar{\omega}$,     the solutions of the system  \eqref{eq:Newton_ES_static_maps1_prescribed-pointwise}-\eqref{eq:Newton_ES_static_maps3_prescribed-pointwise_Hessian} with initial function $\phi:[- 2T,0]  \rightarrow \mathbb{R}^4$  that satisfies $\vert \phi(0) -\left(\theta^\star, 0,{H^\star}^{-1}, H^\star\right)^\top \vert \leq \rho $,   render  the learning dynamics  to satisfy
\begin{equation}\label{eq:main_bound_learning_dynamics_in_theorem}
\left|\hat{\theta}(t)-\theta^\star\right| \leq \nu,\ \forall    t \in [T^{\star} +2D, \bar{T}],   
\end{equation}
 and 
\begin{equation}\label{eq:bound_K_infty_main_theorem}
\left|\hat{\theta}(t)-\theta^\star\right| \leq \gamma\left( \Vert \phi -(\theta^\star, 0,{H^\star}^{-1}, H^\star )^\top \Vert \right) + \nu, \quad \forall t\in[0,\bar{T}].
\end{equation}
\end{theorem}
The proof of Theorem 1 builds on the classical averaging theorem for RFDEs \cite[Section 5]{Hale1990}, focusing on the closeness, over a compact time interval, between the trajectories of the error system and those of its averaged counterpart. In addition, the proof leverages the prescribed-time convergence property of the averaged error learning dynamics.
\begin{IEEEproof}
The nonlinear function $f$ of the RFDE with fast oscillations \eqref{eq:Compact_form_error_system},  \eqref{eq:RFDE-f1}-\eqref{eq:RFDE-f8} is  continuous in its arguments and is continuously differentiable in $X_t$  (notice, in particular, that that  the terms $\mathcal{K}_T(X_5)$, $\mathcal{L}_T(X_5,X_6)$, $\mathcal{K}_{ 2T}(X_7)$ are bounded and differentiable with respect to their arguments). Moreover,  $f$ admits a continuous Fréchet derivative with respect to $X_t$  (which implies $f$  is locally   Lipschitz in $X_t$)  and  also  verifies $f(\omega t + 2\pi,X_t)= f(\omega t ,X_t)$. Consequently, all conditions  on $f$ to apply classical averaging theory in infinite dimensions (see  \cite[Section 5]{Hale1990} and \cite[Section 2.3]{Lehman2002} are met.    In addition, by Lemma \ref{eq:Prt_Target_system} and invoking  the arguments in Remark \ref{remk:forward_invariance_oscillators}, the solution to \eqref{eq:Compact_form_Averaged_error_system} is bounded for $t \geq 0$. Therefore,   applying \cite[Corollary, 5.2]{Hale1990} and \cite[Theorem 4]{Lehman2002}),  we have that for any $\nu >0$ and any arbitrarily large $\bar{T}\geq 0$  (in particular $\bar{T} \gg T^{\star}+ 3D$); there exists $\bar{\omega} >0$  such that for any $\omega > \bar{\omega}$, 
 \begin{equation}\label{eq:estimate_closeness_averaging_inproof}
 \vert  X(t;\varphi)  -  X^{\rm av}(t;\varphi)\vert \leq \nu,   \quad \forall t \in [0,\bar{T}],
 \end{equation}
 where $X(t;\varphi)$ and  $X^{\rm av}(t;\varphi)$ denote the solution of \eqref{eq:Compact_form_error_system}, and \eqref{eq:Compact_form_Averaged_error_system}, respectively.   Therefore, the estimate \eqref{eq:estimate_closeness_averaging_inproof} indicates  i) how close  the solutions of the error system are with respect  to those of its averaged counterpart during the transient phase, and ii) how close the solutions  of the error system  are to those of the  averaged error system once the solutions of the latter have reached the origin in a prescribed time (i.e., at $T^{\star}+2D$).
 
 Therefore, the estimate \eqref{eq:estimate_closeness_averaging_inproof}  holds true, in particular,  for the $\tilde{\theta}$-component of $X$ in  \eqref{eq:Compact_form_error_system}, and the $\tilde{\theta}_{\rm av}$-component of $X^{\rm av}$ in \eqref{eq:Compact_form_Averaged_error_system}. Thus,
  \begin{equation}\label{eq:estimate_closeness_averaging_inproof_first_compt}
 \vert  \tilde{\theta}(t)  -  \tilde{\theta}_{\rm av}(t)\vert \leq \nu,  \quad \forall t \in [0, \bar{T}].
 \end{equation}
 This result along with the triangle inequality  (i.e., using the fact that  $\vert \tilde{\theta}(t) \vert =  \vert \tilde{\theta}(t) + \tilde{\theta}_{\rm av}(t) -\tilde{\theta}_{\rm av}(t)  \vert \leq \vert \tilde{\theta}_{\rm av}(t) \vert + \vert  \tilde{\theta}(t)  -  \tilde{\theta}_{\rm av}(t)\vert  $) yield
\begin{equation}\label{eq:estimate_averaging_inproof_first_compt}
\vert \tilde{\theta}(t) \vert \leq  \vert \tilde{\theta}_{\rm av}(t) \vert + \nu .
\end{equation}
Using Lemma \ref{eq:Prt_Target_system}, the estimate \eqref{eq:estimate_tilde_theta_av_afterLemma1}, and the fact $(\tilde{\theta}^0,z^0,\tilde{\Gamma}^0,\tilde{H}^0)^{\top}=(\tilde{\theta}_{\rm av}^0,z_{\rm av}^0,\tilde{\Gamma}_{\rm av}^0,\tilde{H}_{\rm av}^0)^{\top}=\tilde{\phi}$, we get
\begin{equation}\label{eq:estimate_averaging_inproof_first_compt_final_bound}
\begin{split}
\vert \tilde{\theta}(t) \vert \leq  &\sigma \left(t-T^{\star} - 2D  \right)\Big((1+6T)\exp\left(12 \right)q\left( \Vert (\tilde{\theta}^0,z^0)^{\top}   \Vert \right)  \\
&  \quad  + 36T^2 \exp\left(12 \right)  \sup_{0\leq s\leq \min\{t,4T-\delta_{\xi} 2T\}} (\vert v(s)  \vert) \Big)+ \nu,
\end{split}
\end{equation}
for all $0\leq t\leq \bar{T}$,  with $v(t)$ being the  finite-time vanishing  input,  as in \eqref{eq:vanishing_perturbation_v_in_proof}, i.e.,  $v(t)=-\mathcal{L}_{T}(t) \tilde{\Gamma}_{\rm av}(t)H^{\star}\tilde{\theta}_{\rm av}(t-D)$.  Therefore, we  conclude from \eqref{eq:estimate_averaging_inproof_first_compt_final_bound} that 
\begin{equation}
\left|\hat{\theta}(t)-\theta^\star\right| \leq \nu,\quad \forall  t \in [T^{\star} +2D, \bar{T} ],   
\end{equation}
Finally,  estimate \eqref{eq:bound_K_infty_main_theorem} is an immediate consequence of \eqref{eq:estimate_averaging_inproof_first_compt} and the bound \eqref{eq:bound_Kinfty_lemma} in Lemma \ref{eq:Prt_Target_system}, restricted to the compact time interval $[0,\bar{T}]$.  This concludes the proof.      
\end{IEEEproof}

\subsection{KHV PT Extremum Seeking  with Delay-Free  Map }
The following corollary considers the case when the output map \eqref{outputmap} is delay-free:
\begin{corollary}\label{theo:delay-free_main_ES_algorithm_prescribed-pointwise}
  Let us choose $T^{\star} > 0$ in the KHV PT-ES algorithm \eqref{eq:Newton_ES_static_maps1_prescribed-pointwise}-\eqref{eq:Newton_ES_static_maps3_prescribed-pointwise_Hessian}  with output map \eqref{eq:outout_function_static_maps}, where $D=0$, and with  time-periodic gains $\mathcal{K}_T(t)$,  $\mathcal{L}_{T}(t)$ and $\mathcal{K}_{ 2T}(t)$  defined   in \eqref{eq:mathcalK_T}, \eqref{eq:mathcalL_T} and \eqref{eq:mathcalK_tau}, respectively,  which are generated by the oscillators \eqref{eq:oscillator_system_periodic_time-varying_feedbacks} and \eqref{eq:oscillator_system_periodic_time-varying_feedbacks_tau}, with arbitrary initial conditions in $\mathbb{S}^{1}$ yielding   $\delta_{\zeta}$ and $\delta_{\xi}$  as in  \eqref{eq:definition_of_delta-phase_of_oscillators} and \eqref{eq:definition_of_delta-xi-phase_of_oscillators}, respectively. Let $T$ be selected according to \eqref{eq:characteriation_T}  
with $D=0$, namely
\begin{equation}\label{eq:characteriation_T_for_D=0}
T=\frac{T^{\star}}{4 - \delta_{\zeta}  + 2\Big \lceil \frac{4-(2\delta_{\xi} - \delta_{\zeta}) }{2}\Big\rceil}.
\end{equation}
Then,   there exist $\rho >0$  and  a  class-$\mathcal{K}_{\infty}$ function $\gamma$  such that for  any  $\nu \in (0,\rho)$  and any arbitrarily large $\bar{T}\gg T^*$, there exists $\bar{\omega} >0$  such that for any $\omega > \bar{\omega}$,     the solutions of the system  \eqref{eq:Newton_ES_static_maps1_prescribed-pointwise}-\eqref{eq:Newton_ES_static_maps3_prescribed-pointwise_Hessian} with initial function $\phi:[- 2T,0]  \rightarrow \mathbb{R}^4$  that satisfies $\vert \phi(0) -\left(\theta^\star, 0,{H^\star}^{-1}, H^\star\right)^\top \vert \leq \rho $,    render  the learning dynamics  to satisfy,
\begin{equation}\label{eq:main_bound_learning_dynamics_in_corollary}
\left|\hat{\theta}(t)-\theta^\star\right| \leq \nu,\ \forall t    \in [T^{\star},\bar{T}],    
\end{equation}
and 
\begin{equation}\label{eq:bound_K_infty_main_corollary}
\left|\hat{\theta}(t)-\theta^\star\right| \leq \gamma\left( \Vert \phi -(\theta^\star, 0,{H^\star}^{-1}, H^\star)^\top \Vert \right) + \nu, \quad \forall t\in[0,\bar{T}].
\end{equation}
\end{corollary}
In particular, with  $\mathcal{K}_T(t)$,  $\mathcal{L}_{T}(t)$  and $\mathcal{K}_{ 2T}(t)$ with oscillators initialized with $\zeta(0)=(0,1)^{\top}$, $\xi(0)=(0,-1)^{\top}$ as in Theorem \ref{theo:main_ES_algorithm_prescribed-pointwise}, and any  $T>0$, the statement  \eqref{eq:main_bound_learning_dynamics_in_corollary}  holds with 
\begin{equation}
\label{eq-4T}
T^{\star}=4T\,.
\end{equation}
\begin{remark}
An alternative algorithm for the delay-free map ($D=0$) is given by:
\begin{align}
\dot{\hat{\theta}}(t)=& -\mathcal{K}_{T}(t)\Gamma(t)M(t)y(t-T), \label{eq:Newton_ES_static_maps1_prescribed-pointwise_pure_map_delay}  \\
\dot{\Gamma}(t)=&\mathcal{K}_{ 4T}(t) \Gamma^2(t)\left[\hat{H}(t- 4T)  -  N(t)y(t-T)\right],  
\label{eq:Newton_ES_static_maps3_prescribed-pointwise_Riccati_pure_map_delay} \\
\dot{\hat{H }}(t)=&-\mathcal{K}_{ 4T}(t)\left[\hat{H}(t-  4T)  -N(t)y(t-T)\right], \label{eq:Newton_ES_static_maps3_prescribed-pointwise_Hessian_pure_map_delay}
\end{align}   
with arbitrary $T>0$, and \eqref{eq:dither_signal_as_tiago_S}-\eqref{eq:dither_signal_as_tiago_N} (with $D$ replaced by $T$ in \eqref{eq:dither_signal_as_tiago_S}),
for which, with $\delta_{\zeta}=\delta_{\xi}=0$, the averaged  error $\tilde{\theta}_{\rm av}(t)$ settles to zero in $10T$. Finally,   if the map has a delay $D>0$, the map delay can be leveraged, i.e., $T$ in the algorithm \eqref{eq:Newton_ES_static_maps1_prescribed-pointwise_pure_map_delay}--\eqref{eq:Newton_ES_static_maps3_prescribed-pointwise_Hessian_pure_map_delay} taken as $T=D$, and $y(t-T)$ in \eqref{eq:Newton_ES_static_maps1_prescribed-pointwise_pure_map_delay} replaced by $y(t)$. 
\end{remark}   	
 
\subsection{Discussion  on the main result} \label{discussion_main_results}

In Theorem 1 we have established only a local convergence result. Achieving semi-global convergence within a Newton-based framework remains  challenging.  One potential  direction is to adapt the  ideas from \cite{LabarGaroneKinnaertEbenbauer2019}, together with the techniques in  \cite{FridmanZhang2020,Jbara2025}. This could enable a quantitative analysis and characterization of the parameters involved in the estimation system, relative to an arbitrarily chosen upper bound on the domain of attraction. 

Another feature of our methodology is the reliance on the averaging theorem for infinite-dimensional systems, specifically on the closeness of solutions over a compact time domain i.e., for all $t\in [0,\bar{T}]$ with $\bar{T}$ being arbitrarily larger than the prescribed time. To ensure the closeness of solutions on non-compact time domains -- i.e.,  for all $  t \in [0, +\infty) $ -- it would be necessary to characterize the stability property of the averaged error dynamics via bounds of class-$ \mathcal{KL} $, in the same spirit as it has been done for finite-dimensional systems (e.g., \cite{Teel1999,PovedaKrstic2021}). It is worth noting that the averaged error dynamics \eqref{eq:Newton_ES_static_maps1_prescribed-pointwise_error_averaged}-\eqref{eq:Newton_ES_static_maps3_prescribed-pointwise_Hessian_error_averaged} do exhibit a
prescribed-time convergence property (or dead-beat property of order $T^{\star} +2D$), as established in Lemma \ref{eq:Prt_Target_system},  meaning the solutions reach the origin in finite time and remain there. This  property is  complemented by the existence of a class-$\mathcal{K}_\infty$  bound. However, translating this  bound into the format required by infinite-dimensional averaging theorems--namely, uniform bounds of a class-$\mathcal{KL}$ that hold over infinite time horizons--remains an open problem. This gap prevents us from directly applying averaging results  on non-compact domains and necessitates our restriction to arbitrarily large, but finite, time intervals.

\subsection{KHV PT multivariable extremum seeking with delay-free  map }\label{Multivariable_ES}
Our   results for the scalar case can  be extended to the multivariable case  since, as in \cite{OliveiraKrsticTsubakino2017},  our  Newton-based scheme allows also a diagonal structure -- with a time-periodic delayed feedback applied for each channel. For simplicity, we consider a multivariable static map  without delay.  
Recall that  the goal  is to adjust the input $\theta=(\theta_1,\cdots,\theta_n)^{\top} \in\mathbb{R}^n$ in real time,  in order to drive and maintain the output
\begin{equation}
y(t)=Q(\theta(t)),
\end{equation}
 around the unknown extremum $y^{\star} =Q(\theta^\star)$, where  $\theta^\star \in \mathbb{R}^n$ is the unknown optimizer. For maximum seeking purposes, we assume  
\begin{equation}\label{eq:Gradient_Hessian_equilibrium_MULTIVAR}
\frac{\partial Q(\theta^\star)}{\partial \theta}=0,  \quad \frac{\partial^2 Q(\theta^\star)}{\partial \theta^2}=H^\star <0, \quad \text{and} \quad H^{\star}=H^{\star \top},
\end{equation}
where the Hessian $H^\star \in \mathbb{R}^{n\times n}$ of the static map  is unknown. The output $y$ is assumed quadratic, 
\begin{equation}\label{eq:outout_function_static_maps_MULTIVAR}
y(t)=  y^\star +  \frac{1}{2}\left(\theta(t) - \theta^\star \right)^{\top}H^{\star}\left(\theta(t) - \theta^\star \right). 
\end{equation}
The dither signals $S(t)$, $M(t)$, and   $N(t)$ are given as  \cite{GhaffariKrsticNesic2012}:
\begin{align}
S(t)&=(a_1\sin(\omega_1 t),\cdots,a_n\sin(\omega_nt) )^{\top}, \label{eq:dither_signal_as_tiago_S_MULTIVAR}\\ 
M(t) &=\left(\tfrac{2}{a_1}\sin(\omega_1 t),\cdots,\tfrac{2}{a_n}\sin(\omega_n t) \right)^{\top},\label{eq:dither_signal_as_tiago_M_MULTIVAR}
\end{align} 
and $N(t)$,  a matrix given by
\begin{align}\label{eq:dither_signal_as_tiago_N_MULTIVAR}
N_{ii}(t)&=\tfrac{16}{a_{i}^2}\left(\sin^2(\omega_it) - \tfrac{1}{2} \right) \quad \text{for} \quad i=1,2,\cdots,n,\\
N_{ij}(t)&=\tfrac{4}{a_ia_j}\sin(\omega_i t)\sin(\omega_j t)  \quad \text{for} \quad i\neq j,
\end{align}
with  nonzero perturbation  amplitudes  $a_i$, and frequencies	 $\omega_i>0$ chosen to satisfy
$\omega_i=\omega_i^{'}\omega=\mathcal{O}(\omega), \quad i=1,2, \cdots,n$,
where  $\omega$  is a positive constant and $\omega_i^{'}$ a rational number, verifying,  for all distinct $i,j,k$, and $l$ \cite{GhaffariKrsticNesic2012}:
\begin{equation}
\omega_i^{'} \notin \{ \omega_j^{'},  \tfrac{1}{2}(\omega_j^{'}+\omega_k^{'}), \omega_j^{'}+2\omega_{k}^{'},\omega_{j}^{'}+\omega_{k}^{'} \pm \omega_l^{'} \}.
\end{equation} 
The averaging properties \eqref{property:averaging_get_gradient}-\eqref{property:averaging_get_Hessian} hold for the multivariable case, with  $\Pi= 2\pi \times \text{LCM}\lbrace  \tfrac{1}{\omega_i} \rbrace$, $i=1,2,\cdots,n$, where LCM stands for the least common multiple. 
 
We propose the following KHV PT multivariable ES\footnote{In \eqref{eq:Newton_ES_static_maps3_prescribed-pointwise_Riccati_MULTIVAR} and \eqref{eq:Newton_ES_static_maps3_prescribed-pointwise_Hessian_MULTIVAR}, the states are matrices of dimension $n\times n$. Therefore, the dynamics of $\Gamma(t)$ and $\hat{H}(t)$ must be understood as a matrix RFDE. This notation is used to simplify our presentation.}:
\begin{align}
\dot{\hat{\theta}}(t)=&  z(t),  \label{eq:Newton_ES_static_maps1_prescribed-pointwise_MULTIVAR} \\
\dot{z}(t) =&-2\mathcal{K}_{T}(t)z(t-T)
\nonumber\\ &
+\mathcal{L}_{T}(t) \left[ \hat{\theta}(t) - \hat{\theta}(t-T)
-\Gamma(t)M(t)y(t)  \right]
,  \label{eq:Newton_ES_static_maps2_prescribed-pointwise_MULTIVAR} \\
\dot{\Gamma}(t)=&\mathcal{K}_{ 2T}(t) \Gamma(t)\left[\hat{H}(t- 2T)  -  N(t)y(t)\right]\Gamma(t),
\label{eq:Newton_ES_static_maps3_prescribed-pointwise_Riccati_MULTIVAR} \\
\dot{\hat{H }}(t)=&-\mathcal{K}_{ 2T}(t)\left[\hat{H}(t-  2T)  -N(t)y(t)\right], \label{eq:Newton_ES_static_maps3_prescribed-pointwise_Hessian_MULTIVAR}
\end{align}   
where the static map $y(t)$ is now given in \eqref{eq:outout_function_static_maps_MULTIVAR}, with actual input  $\theta(t)=\hat{\theta}(t)+S(t)$,   $\hat{\theta}(t)\in\mathbb{R}^{n}$ is the learning dynamics, $z(t)\in\mathbb{R}^{n}$,   $\Gamma(t)\in\mathbb{R}^{n \times n}$,  and $\hat{H} \in \mathbb{R}^{n\times n}$.  Recall that $\Gamma = \hat{H}^{-1}$ and that \eqref{eq:Newton_ES_static_maps3_prescribed-pointwise_Riccati_MULTIVAR} is obtained using  $\dot{\Gamma}= - \Gamma \dot{\hat{H}} \Gamma$. The periodic gains $\mathcal{K}_T(t)$ and $\mathcal{L}_{T}(t)$, and $\mathcal{K}_{ 2T}(t)$  are defined, respectively  in \eqref{eq:mathcalK_T}, \eqref{eq:mathcalL_T} and \eqref{eq:mathcalK_tau}.  
  We consider  initial conditions   $\hat{\theta}^0=\phi_{\hat{\theta}}(s)$ with $\phi_{\hat{\theta}} \in C^{0}([- 2T,0]; \mathbb{R}^{n}) $, $z^0=\phi_{z}(s)$ with $\phi_z \in C^{0}([- 2T,0]; \mathbb{R}^{n}) $,   $\hat{H}^0=\phi_{\hat{H}}(s)$, with $\phi_{\hat{H}} \in C^{0}([- 2T,0]; \mathbb{R}^{n \times n})$, and $\hat{H}^0 \neq 0$ for all $s\in [- 2T,0] $,  and $\Gamma^0 =(\hat{H}^0)^{-1} = \phi_{\Gamma}(s) \in C^{0}([- 2T,0]; \mathbb{R}^{n\times n})$,  for \eqref{eq:Newton_ES_static_maps1_prescribed-pointwise_MULTIVAR}-\eqref{eq:Newton_ES_static_maps3_prescribed-pointwise_Hessian_MULTIVAR},    $\zeta(0)=(\zeta_1(0),\zeta_2(0))^{\top} \in \mathbb{S}^1$ for  \eqref{eq:oscillator_system_periodic_time-varying_feedbacks}, and $\xi(0)=(\xi_1(0),\xi_2(0))^{\top} \in \mathbb{S}^1$ for \eqref{eq:oscillator_system_periodic_time-varying_feedbacks_tau}. We state next the main result for the multivariable case. The proof follows the same arguments as in the proof of Theorem \ref{theo:main_ES_algorithm_prescribed-pointwise}, thus it is omitted.
 \begin{theorem}\label{theo:main_ES_algorithm_prescribed-pointwise_MULTIVAR}
  Let $T^{\star} >0$ be the prescribed time.
Consider the KHV PT multivariable ES algorithm \eqref{eq:Newton_ES_static_maps1_prescribed-pointwise_MULTIVAR}-\eqref{eq:Newton_ES_static_maps3_prescribed-pointwise_Hessian_MULTIVAR}  with output map \eqref{eq:outout_function_static_maps_MULTIVAR}  and  time-periodic gains $\mathcal{K}_T(t)$,  $\mathcal{L}_{T}(t)$ and $\mathcal{K}_{ 2T}(t)$  defined   in \eqref{eq:mathcalK_T}, \eqref{eq:mathcalL_T} and \eqref{eq:mathcalK_tau}, respectively,  and which are generated by the oscillators \eqref{eq:oscillator_system_periodic_time-varying_feedbacks} and \eqref{eq:oscillator_system_periodic_time-varying_feedbacks_tau}, with arbitrary initial conditions in $\mathbb{S}^{1}$ yielding   $\delta_{\zeta}$ and $\delta_{\xi}$  as in  \eqref{eq:definition_of_delta-phase_of_oscillators} and \eqref{eq:definition_of_delta-xi-phase_of_oscillators}, respectively. Let $T$ be selected according to
\begin{equation}\label{eq:characteriation_T_for_D=0_MULTIVAR}
T=\frac{T^{\star}}{4 - \delta_{\zeta}  + 2\Big \lceil \frac{4-(2\delta_{\xi} - \delta_{\zeta}) }{2}\Big\rceil}.
\end{equation}
Then, there exists $\rho >0$ such that  for any $\nu \in (0,\rho)$  and any arbitrary large $\bar{T}\gg T^*$, there exists $\bar{\omega} >0$  such that for any $\omega > \bar{\omega}$,     the solutions of the system  \eqref{eq:Newton_ES_static_maps1_prescribed-pointwise_MULTIVAR}-\eqref{eq:Newton_ES_static_maps3_prescribed-pointwise_Hessian_MULTIVAR} with  initial functions  $\phi_{\hat{\theta}}:[- 2T,0]  \rightarrow \mathbb{R}^n$, $\phi_{z}:[- 2T,0]  \rightarrow \mathbb{R}^n$ $\phi_{\Gamma}:[- 2T,0]  \rightarrow \mathbb{R}^{n\times n}$, $\phi_{\hat{H}}:[- 2T,0]  \rightarrow \mathbb{R}^{n\times n}$  that satisfy  $\vert \phi_{\hat{\theta}}(0) - \theta^\star \vert \leq \rho$, $\vert \phi_{z}(0)  \vert \leq \rho$, $\vert \phi_{\Gamma}(0) - {H^\star}^{-1}\vert \leq \rho$ and  $\vert \phi_{\hat{H}}(0) - H^\star \vert \leq \rho $,   render  the learning dynamics  to satisfy, 
\begin{equation}\label{eq:main_bound_learning_dynamics_in_theorem_MULTIVAR}
\left|\hat{\theta}(t)-\theta^\star\right|  \leq \nu,\ \forall t \in  [T^{\star}, \bar{T}].    
\end{equation}
 \end{theorem}

\section{Numerical Simulations}\label{numerical_simulations}

\subsection{Scalar static map with delayed measurement}
We consider the following scalar static
quadratic map with the delayed measurement (with $D=5s$):
\begin{align}
y(t)=& Q(\theta(t-5)) \quad  Q(\theta)=5-0.4(\theta - 2)^2
\end{align}
We  implement the KHV PT-ES algorithm \eqref{eq:Newton_ES_static_maps1_prescribed-pointwise}-\eqref{eq:Newton_ES_static_maps3_prescribed-pointwise_Hessian}. We use the dither signals given  by 
\eqref{eq:dither_signal_as_tiago_S}-\eqref{eq:dither_signal_as_tiago_N} whose parameters  are $a=0.5$, $\omega=1000$. 
 We consider the oscillator \eqref{eq:oscillator_system_periodic_time-varying_feedbacks} with initial conditions $\zeta(0)=(0,1)^{\top}$  and the oscillator \eqref{eq:oscillator_system_periodic_time-varying_feedbacks_tau} with initial conditions $\xi(0)=(0,1)^{\top}$, yielding in both cases, $\delta_{\zeta}=0$ and $\delta_{\xi}=0$ computed according to \eqref{eq:definition_of_delta-phase_of_oscillators}, and  \eqref{eq:definition_of_delta-xi-phase_of_oscillators}, respectively.  
We set the prescribed time $T^{\star}=40s$ (recall $T^{\star}>  6D$  according to \eqref{eq:restriction_choice_T_star}), and we select $T=6.25s$  according to  \eqref{eq:characteriation_T}. 
We set the initial state $\hat{\theta}(0)=-5$  of the estimator and its preceding delay values (i.e., $\hat{\theta}(s)=-5$ for all $s \in [-2T,0]$ knowing that $T> D$). We set initial state  of the filter $z$  as  $z(0)=0.1$ and its preceding delay values (i.e., $z(s)=0.1$ for all   $s \in  [-2T,0]$)  and we set  $\Gamma(0)=-0.9$ (thus $\hat{H}(0)=-1.11$ and its preceding delay values i.e., $\hat{H}(s)=-1.11$ for all  $s \in [- 2T,0]$). 
\begin{figure}[t!]
\centering{
\includegraphics[width=0.88\columnwidth]{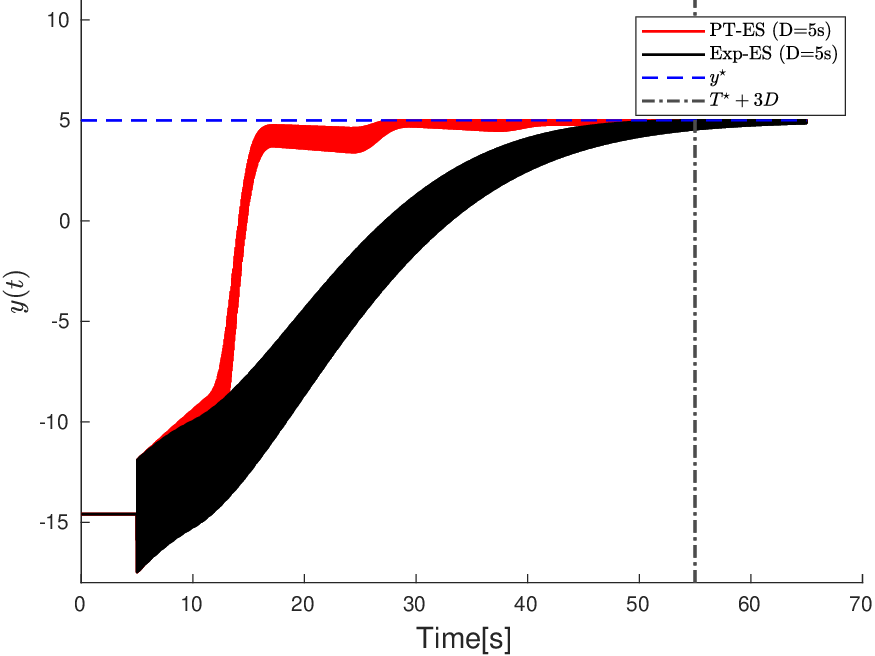} 
\caption{Time-evolution of the static map \eqref{eq:outout_function_static_maps}  under the KHV PT-ES algorithm \eqref{eq:Newton_ES_static_maps1_prescribed-pointwise}-\eqref{eq:Newton_ES_static_maps3_prescribed-pointwise_Hessian} with  prescribed time of convergence is dictated by  $T^{\star}+ 3D = 55s$ (red line) and under the exponential  Newton-based extremum seeking algorithm \eqref{eq:Newton_ES_static_maps1V2}-\eqref{eq:Newton_ES_static_maps2V2} (black line). }
\label{Static-map_and_Estimators_plot}
}
\end{figure}

\begin{figure*}[t]
\centering{
\subfigure{\includegraphics[width=0.88\columnwidth]{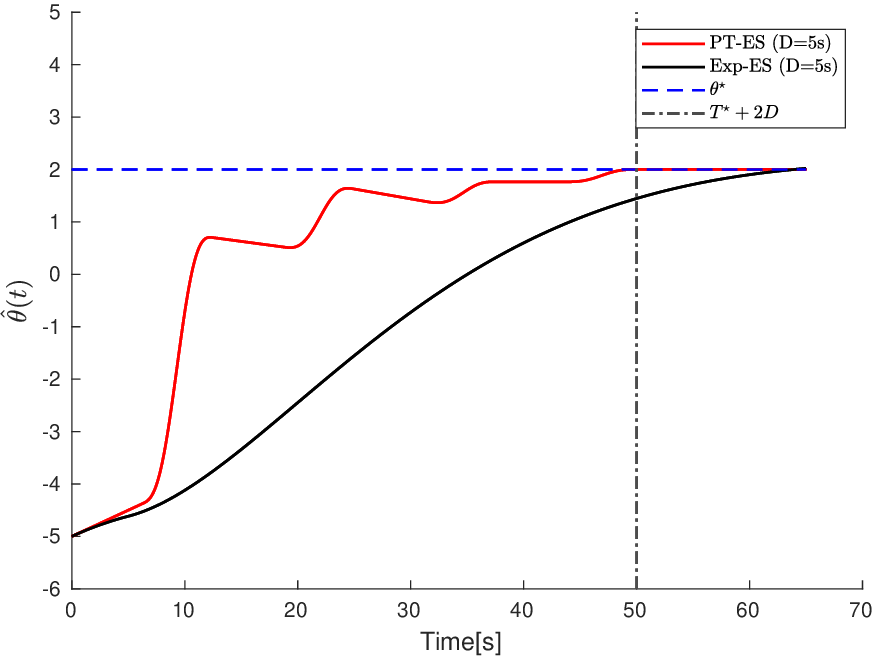}}
\subfigure{\includegraphics[width=0.9\columnwidth]{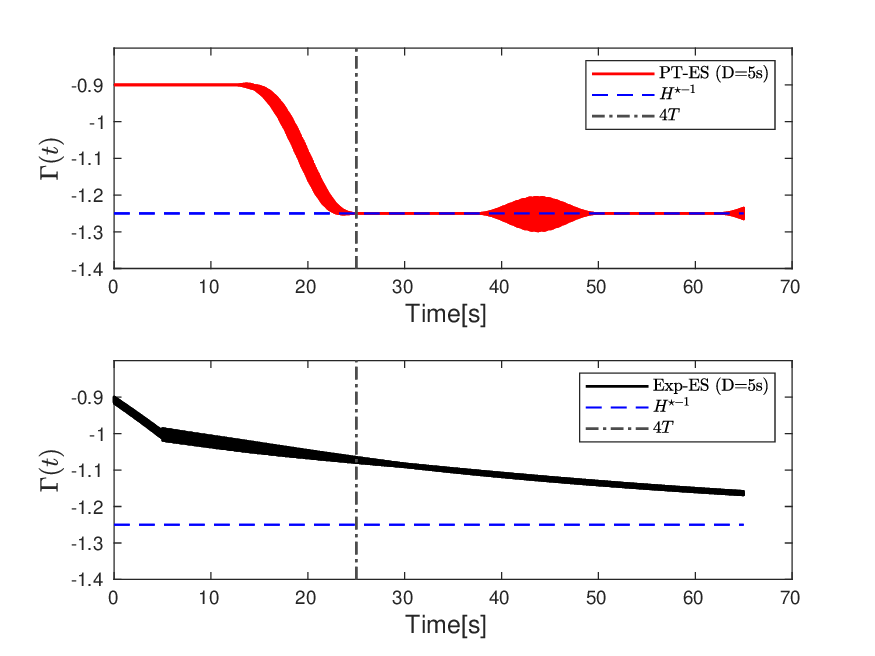} } 
\caption{ \underline{Left:} Time-evolution of the estimator $\hat{\theta}(t)$ under the KHV PT-ES algorithm \eqref{eq:Newton_ES_static_maps1_prescribed-pointwise}-\eqref{eq:Newton_ES_static_maps3_prescribed-pointwise_Hessian} with   prescribed-time   of convergence    $T^{\star}+2D=50 s$  (red line), and under  the exponential    ES algorithm \eqref{eq:Newton_ES_static_maps1V2}-\eqref{eq:Newton_ES_static_maps2V2} (black line). \underline{Right:} Time-evolution of the Riccati-like filter estimating the inverse of the Hessian  under the KHV PT-ES algorithm \eqref{eq:Newton_ES_static_maps1_prescribed-pointwise}-\eqref{eq:Newton_ES_static_maps3_prescribed-pointwise_Hessian} with prescribed time of convergence $ 4T = 25s$ (top--red line), and under  the exponential   ES algorithm \eqref{eq:Newton_ES_static_maps1V2}-\eqref{eq:Newton_ES_static_maps2V2} (bottom--black line).  }
\label{Learning_dynamics_plots}
}
\end{figure*}

\begin{figure}[t]
\centering{
\includegraphics[width=0.88\columnwidth]{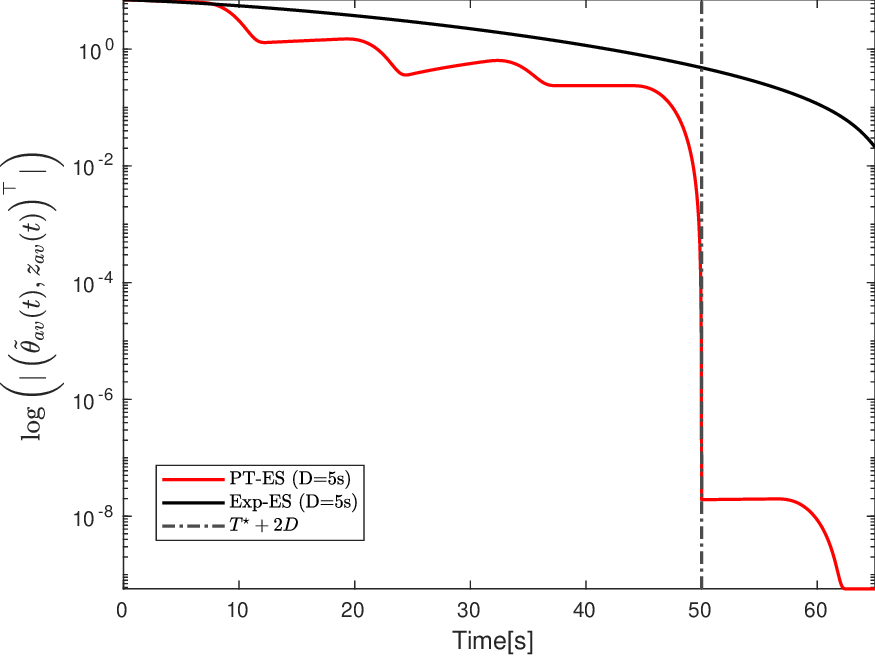} 
\caption{Plot in logarithmic scale  of the Euclidean norm of the averaged error dynamics   \eqref{eq:Newton_ES_static_maps1_prescribed-pointwise_error_averaged}-\eqref{eq:Newton_ES_static_maps2_prescribed-pointwise_error_averaged} (red line)  and \eqref{eq:Newton_ES_static_averaged1}-\eqref{eq:Newton_ES_static_averaged2}   (black line).}
\label{Error_dynamics_plots_logaritmic_scale}
}
\end{figure}

Figure \ref{Static-map_and_Estimators_plot} shows on the left the evolution of the output $y(t)$ with delay $D=5s$. we use: i) the KHV PT-ES algorithm \eqref{eq:Newton_ES_static_maps1_prescribed-pointwise}-\eqref{eq:Newton_ES_static_maps3_prescribed-pointwise_Hessian} (red line)    and  ii)  the exponential  Newton-based extremum seeking algorithm \eqref{eq:Newton_ES_static_maps1V2}-\eqref{eq:Newton_ES_static_maps2V2} (black line). For the latter ES algorithm, we have selected $w_r=0.02$, $c=0.1$ and $K=0.05$. 
Figure \ref{Learning_dynamics_plots}  shows the evolution of $\hat{\theta}(t)$ (left) and  the $\Gamma(t)$ (right)  by using: i) the KHV PT-ES \eqref{eq:Newton_ES_static_maps1_prescribed-pointwise}-\eqref{eq:Newton_ES_static_maps3_prescribed-pointwise_Hessian} (red line);    and ii)  the exponential  Newton-based ES algorithm \eqref{eq:Newton_ES_static_maps1V2}-\eqref{eq:Newton_ES_static_maps2V2} (black line). 

Finally, Figure \ref{Error_dynamics_plots_logaritmic_scale}  shows the Euclidean norm, in logarithmic scale,  of the solution of the averaged error dynamics $\tilde{\theta}_{\rm av}(t)$ and $z_{\rm av}(t)$.  One can observe (red line) the convergence  to zero of the solutions to  \eqref{eq:Newton_ES_static_maps1_prescribed-pointwise_error_averaged}- \eqref{eq:Newton_ES_static_maps2_prescribed-pointwise_error_averaged} within a prescribed-time   $T^{\star}+2D=50 s$.

\subsection{Multivariable static map with delay-free map}
 We borrow the example in \cite[Section IX]{OliveiraKrsticTsubakino2017} for the output delay-free case. Hence,  the output function has the form of \eqref{eq:outout_function_static_maps_MULTIVAR} with the local maximum  $y^{\star}=1$ and the unknown maximizer $\theta^{\star}= (\theta_1^{\star}, \theta_2^{\star})^{\top}=(0,1)^{\top}$; and for which we assume the Hessian is given by 
$H^{\star}=\begin{bmatrix}
-2 & -2 \\
-2 & -4
\end{bmatrix}$. 
We  implement the KHV PT multivariable ES algorithm \eqref{eq:Newton_ES_static_maps1_prescribed-pointwise_MULTIVAR}-\eqref{eq:Newton_ES_static_maps3_prescribed-pointwise_Hessian_MULTIVAR}. We use the dither signals given  by 
\eqref{eq:dither_signal_as_tiago_S_MULTIVAR}-\eqref{eq:dither_signal_as_tiago_N_MULTIVAR} whose parameters  are $a_1=a_2=0.5$, $\omega_1=70$, $\omega_1=50\omega$, $\omega_2=70 \omega$.   We consider the oscillator \eqref{eq:oscillator_system_periodic_time-varying_feedbacks} with initial conditions $\zeta(0)=(0,1)^{\top}$  and the oscillator \eqref{eq:oscillator_system_periodic_time-varying_feedbacks_tau} with initial conditions $\xi(0)=(0,1)^{\top}$, yielding in both cases, $\delta_{\zeta}=0$ and $\delta_{\xi}=0$ computed according to \eqref{eq:definition_of_delta-phase_of_oscillators}, and  \eqref{eq:definition_of_delta-xi-phase_of_oscillators}, respectively.  
We set the prescribed time $T^{\star}=30s$ and we select $T=3.75s$  according to  \eqref{eq:characteriation_T_for_D=0_MULTIVAR}. 
We set the initial state $\hat{\theta}(0)=(-1,2)^{\top}$  of the estimator and its preceding delay values. We set initial state  of the filter $z$  as  $z(0)=(0,0)^{\top}$ and its preceding delay values,  and we set  $\Gamma(0)=-\text{diag}\{0.01, 0.005\}$  as in   \cite{OliveiraKrsticTsubakino2017} (thus $\hat{H}(0)=-\text{diag}\{100,200 \}$ and its preceding delay values). 

\begin{figure*}[t]
\centering{
\subfigure{\includegraphics[width=0.88\columnwidth]{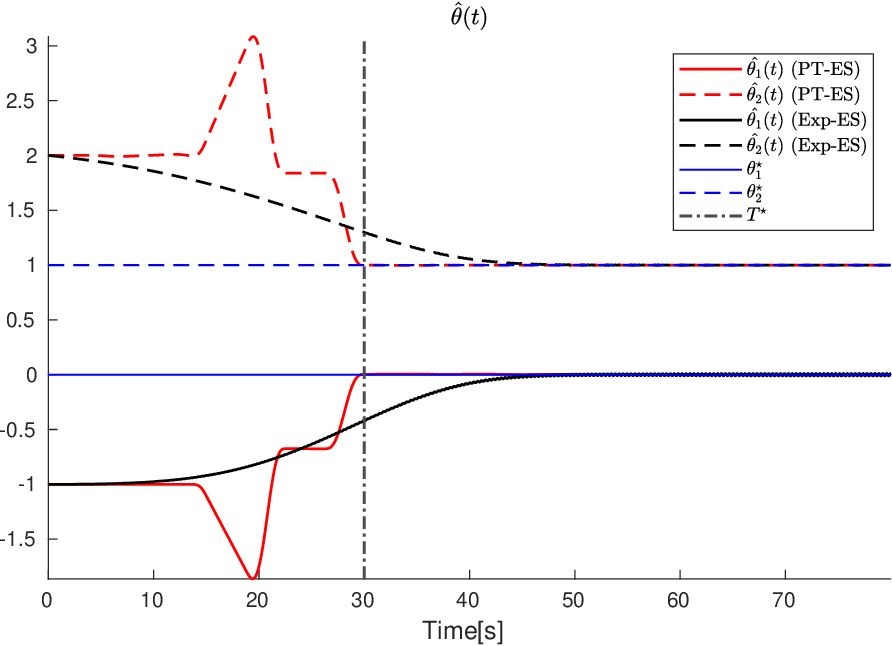}}
\subfigure{\includegraphics[width=0.9\columnwidth]{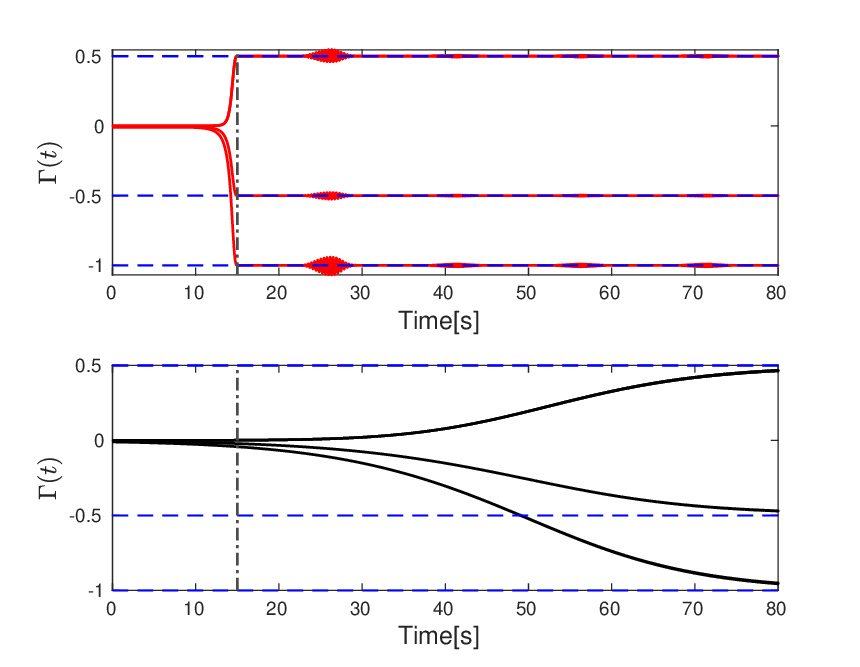} } 
\caption{{\color{black} \underline{Left:} Time-evolution of the estimator $\hat{\theta}(t)$ under the KHV PT multivariable ES algorithm \eqref{eq:Newton_ES_static_maps1_prescribed-pointwise_MULTIVAR}-\eqref{eq:Newton_ES_static_maps3_prescribed-pointwise_Hessian_MULTIVAR} with   prescribed-time   of convergence    $T^{\star}=30 s$  (red lines), and under  the exponential  multivariable  ES algorithm in \cite[Section VII, eq (123)]{OliveiraKrsticTsubakino2017} (black lines). \underline{Right:} Time-evolution of the Riccati-like filter estimating the inverse of the Hessian  under the KHV PT multivariable ES algorithm \eqref{eq:Newton_ES_static_maps1_prescribed-pointwise_MULTIVAR}-\eqref{eq:Newton_ES_static_maps3_prescribed-pointwise_Hessian_MULTIVAR} with prescribed time of convergence $ 4T = 15s$ (top--red lines), and under  the exponential  multivariable  ES algorithm in \cite[Section VII, eq (123)]{OliveiraKrsticTsubakino2017} (bottom--black lines).  }}
\label{Learning_dynamics_plots_MULTIVAR}
}
\end{figure*}

Figure  \ref{Learning_dynamics_plots_MULTIVAR}  shows the evolution of  $\hat{\theta}(t)$ (left) and   $\Gamma(t)$ (right), respectively,  if  we use: i) the KHV PT multivariable ES algorithm \eqref{eq:Newton_ES_static_maps1_prescribed-pointwise_MULTIVAR}-\eqref{eq:Newton_ES_static_maps3_prescribed-pointwise_Hessian_MULTIVAR} (red lines)    and  ii)  the exponential  multivariable  ES algorithm in \cite[Section VII, eq (123)]{OliveiraKrsticTsubakino2017} (black lines)  with tuning parameters  $w_r=0.1$, $c_1=c_2=20$ and $K=\text{diag}\{1,1 \}$. Under the the KHV PT multivariable ES algorithm  the time of convergence for  the learning dynamics $\hat{\theta}(t)$ to estimate $\theta^\star$ is given by  $T^{\star} =30s$ and the time of convergence for $\Gamma(t)$  to estimate the Hessian's inverse is given by  $ 4T = 15s$.

\section{Conclusion}\label{Section:Conclusion_ExTS}
This paper presented a novel  Newton-like extremum seeking (ES) algorithm for static maps subject to time delay. The proposed algorithm achieves delay-compensation alongside convergence towards the optimum in a finite time  that can be prescribed in the design. It builds on time-periodic gains and delayed feedbacks.  We have  refereed to it as the {\em KHV  PT-ES}.  We extended the result to  multivariable static maps (with delay-free map).
The key advantage of this approach over existing methods is its ability to achieve PT convergence without relying on singular gains. 

Future work will focus on extending the KHV PT-ES methodology to dynamical systems with delays, as well as addressing  the  challenges outlined in Section \ref{discussion_main_results}.

\appendix 
\subsection{Proof of Lemma \ref{Lemma_PrT_Pointwise_delay}}\label{proof_of_lemma1}
\begin{proof}
The one-dimensional control system \eqref{eq:single_integrator}, \eqref{eq:time-varying_delayed-feedback_key_result} reads as follows:
\begin{equation}\label{eq:single_integrator_closed-loop_in_proof}
\dot{x}(t)=-\mathcal{K}_{T}\left(t\right)x(t-T) + v(t).
\end{equation}
The solution to \eqref{eq:single_integrator_closed-loop_in_proof}, with initial condition $x_0=\phi(s) $, $\phi \in C^0([-T,0];\mathbb{R})$ is given  by $x(t)=\phi(t)$ for $t \in [-T,0]$; and
\begin{equation}\label{eq:direct_integration_solution_closed_loopsystem}
x(t)=\phi(0)  - \int_{0}^{t}\mathcal{K}_{T}(s)x(s-T)ds + \int_{0}^{t}v(s)ds,
\end{equation}
 for $t \geq 0$. 
 Hence, by  \cite[Chapter 6]{Hale1993}, we can obtain the following estimate for \eqref{eq:direct_integration_solution_closed_loopsystem}:
 \begin{equation}
 \begin{split}
 \vert x(t) \vert \leq& \max_{-T \leq s\leq 0}(\vert \phi(s) \vert) + \int_{-T}^{0}\vert \mathcal{K}_{T}(r+T)\vert \vert \phi(r) \vert dr  \\
 &+ \int_{0}^{t}\vert \mathcal{K}_{T}(r+T)\vert \vert x(r) \vert dr + \int_{0}^{t}\vert v(s) \vert ds.
\end{split} 
 \end{equation}
 The definition of $\mathcal{K}_{T}$  ensures  that $\mathcal{K}_{T}(r +T)=0$ for $r  \in [-T,0]$, and that $\vert \mathcal{K}_{T}(r+T) \vert \leq \tfrac{2}{T}$ for $r  \in [0,t]$. Therefore we obtain, 
 \begin{equation*}
 \vert x(t) \vert \leq \max_{-T \leq s\leq 0}(\vert \phi(s) \vert) + \int_{0}^{t}\tfrac{2}{T} \vert x(r) \vert dr + \int_{0}^{t}\vert v(s) \vert ds,
 \end{equation*}
 for $t \geq 0$.  Using the Grönwall-Bellman inequality, we further obtain
 \begin{equation}\label{eq:exponential_estimate_withGronwall}
 \vert x(t) \vert \leq \left( \max_{-T \leq s\leq 0}(\vert \phi(s) \vert) + t \sup_{0\leq s\leq t}  (\vert v(s) \vert \right) \exp\left( \tfrac{2}{T}t \right). 
 \end{equation}
 Consider now the time sequence  $\bar{t}_n=2Tn-\delta_{\zeta} T$, $\delta_{\zeta} \in (-1,1]$, $n\in\mathbb{Z}^{+}$. Then,  for $t \in [\bar{t}_n,\bar{t}_n +T]$, the solution to \eqref{eq:single_integrator_closed-loop_in_proof} is given as follows:
\begin{equation}\label{eq:solution_closed_loopsystem_interval_Ktequal_0}
x(t) = x(\bar{t}_n) + \int_{\bar{t}_n}^tv(\tau)d\tau,
\end{equation}
from which we can obtain the following estimate:
\begin{equation}\label{eq:first_estimate_interval}
\vert x(t) \vert \leq \vert x(\bar{t}_n)\vert  + T \sup_{\bar{t}_n\leq \tau \leq t}(\vert v(\tau) \vert).
\end{equation}
For  $t \in [\bar{t}_n + T,\bar{t}_{n+1}]$, the solution to \eqref{eq:single_integrator_closed-loop_in_proof} is given as follows:
\begin{equation}
\begin{split}
x(t) = & x(\bar{t}_n +T) + \int_{\bar{t}_n+T}^tv(\tau)d\tau - \int_{\bar{t}_n +T}^t \mathcal{K}_{T}(\tau)x(\tau-T)d\tau. \\
\end{split}
\end{equation}
Using \eqref{eq:solution_closed_loopsystem_interval_Ktequal_0}, we further obtain, for all $t \in [\bar{t}_n +T, \bar{t}_{n+1}]$
\begin{equation*}
\begin{split}
x(t)=& x(\bar{t}_n) - \int_{\bar{t}_n +T}^t \mathcal{K}_{T}(\tau)x(\bar{t}_n)d\tau + \int_{\bar{t}_n}^{\bar{t}_n+T}v(\tau)d\tau \\
& + \int_{\bar{t}_n+T}^tv(\tau)d\tau    - \int_{\bar{t}_n +T}^t \mathcal{K}_{T}(\tau)\left(\int_{\bar{t}_n}^{\tau-T} v(s)ds \right)d\tau. 
\end{split}
\end{equation*}
Therefore, we get
\begin{equation}\label{eq:explicit_solution_on_interval_inproof}
\begin{split}
x(t)=& x(\bar{t}_n) \bar{K}_{T}(\bar{t}_n,t) \\
&+ \int_{\bar{t}_n}^{t}v(\tau)d\tau     - \int_{\bar{t}_n +T}^t \mathcal{K}_{T}(\tau)\left(\int_{\bar{t}_n}^{\tau-T} v(s)ds \right)d\tau. 
\end{split}
\end{equation}
with 
\begin{equation}\label{eq:K_T_bar_in-proof_lemma1}
\bar{K}_{T}(\bar{t}_n,t)=1 - \int_{\bar{t}_n +T}^t \mathcal{K}_{T}(\tau)d\tau.
\end{equation} 
Hence, from \eqref{eq:explicit_solution_on_interval_inproof}, we obtain the following estimate, for all $t \in [\bar{t}_n+T, \bar{t}_{n+1}]$:
\begin{equation*}
\begin{split}
\vert x(t) \vert \leq & \vert x(\bar{t}_n) \vert \vert \bar{K}_{T}(\bar{t}_n,t) \vert + (t-\bar{t}_n)\sup_{\bar{t}_n \leq \tau \leq t} (\vert v(\tau) \vert) \\
& +(t-T -\bar{t}_n)\sup_{\bar{t}_n \leq \tau \leq t}(\vert v(\tau) \vert)\left( \int_{\bar{t}_n +T}^t \mathcal{K}_{T}(\tau)d\tau \right). 
\end{split}
\end{equation*}
Thus, 
\begin{equation}\label{eq:estimate_x(t)_in_proof_on_interval}
\vert x(t) \vert \leq  \vert x(\bar{t}_n) \vert \vert \bar{K}_{T}(\bar{t}_n,t) \vert + 3T\sup_{\bar{t}_n \leq \tau \leq t}(\vert v(\tau) \vert),
\end{equation}
where we have used $\int_{\bar{t}_n+T}^{t}\mathcal{K}_{T}(\tau)d\tau \leq \int_{\bar{t}_n+T}^{\bar{t}_{n+1}}\mathcal{K}_{T}(\tau)d\tau $ and  the property 
\begin{equation}\label{eq:Property_integral_K_T}
\int_{\bar{t}_n+T}^{\bar{t}_{n+1}}\mathcal{K}_{T}(\tau)d\tau = 1.
\end{equation}
This property can be verified by recalling the definition \eqref{eq:mathcalK_T_V0} in conjunction with \eqref{eq:explicit_solution_osicillator1_V2}. 
Notice that $\vert x(\bar{t}_n) \vert$ in  \eqref{eq:first_estimate_interval} or in  \eqref{eq:estimate_x(t)_in_proof_on_interval} can further be upper bounded using \eqref{eq:exponential_estimate_withGronwall}. 
Moreover, from \eqref{eq:estimate_x(t)_in_proof_on_interval}, at $t=\bar{t}_{n+1}$, and using the fact that $\bar{K}_{T}(\bar{t}_n,\bar{t}_{n+1})=0$ due to   \eqref{eq:K_T_bar_in-proof_lemma1} and \eqref{eq:Property_integral_K_T}, the following estimate holds:
\begin{equation}\label{eq:estimate_at_tnplus1}
\vert x(\bar{t}_{n+1}) \vert \leq 3T\sup_{\bar{t}_n \leq \tau \leq \bar{t}_{n+1}} (\vert v(\tau) \vert).
\end{equation}
Repeating the same reasoning on the interval $ [\bar{t}_{n+1}, \bar{t}_{n+2}]$ (the control features zero gain,  acting on the zero solution), we can get
\begin{equation}\label{eq:estimate_at_tnplus2}
\vert x(t) \vert \leq 6T \sup_{\bar{t}_n \leq \tau \leq t}  (\vert v(\tau) \vert).
\end{equation}
Hence, 
using the fact that $t-(2T-\delta_{\zeta} T) \leq \bar{t}_n$, for any $n \in \mathbb{Z}^{+}$
 and   combining \eqref{eq:exponential_estimate_withGronwall}  {\color{black} (for the interval $[\max\{0,\bar{t}_0 \},\bar{t}_1]$)  }, \eqref{eq:first_estimate_interval},  \eqref{eq:estimate_x(t)_in_proof_on_interval},  \eqref{eq:estimate_at_tnplus1} {\color{black} (for the interval $[\bar{t}_1 ,\bar{t}_2]$)} and \eqref{eq:estimate_at_tnplus2} {\color{black} (for the interval $[\bar{t}_2 ,+\infty)$)} we finally obtain  \eqref{eq:bound_PrT_single_integrator}.

  To show the  last part of the conclusion,    we follow  similar steps  as those  below \eqref{eq:exponential_estimate_withGronwall} and distinguish two cases:\\
 i) $v(t)=0$ for $t\geq t^{'}$ where $ \bar{t}_n \leq t^{'} \leq  \bar{t}_n +T$. We have that for   $t \in [\bar{t}_n,\bar{t}_n +T]$, the solution to \eqref{eq:single_integrator_closed-loop_in_proof} is given as follows:
\begin{equation}\label{eq:solution_closed_loopsystem_interval_Ktequal_0-distrubance}
x(t) = x(\bar{t}_n) + \bar{v}(\bar{t}_n,t),
\end{equation}
and, for   $t \in [\bar{t}_n + T,\bar{t}_{n+1}]$, the solution to \eqref{eq:single_integrator_closed-loop_in_proof} is given as follows:
\begin{equation}\label{eq:solution_closed_loopsystem_interval_Ktequal_1-distrubance}
\begin{split}
x(t) = & x(\bar{t}_n +T)  - \int_{\bar{t}_n +T}^t \mathcal{K}_{T}(\tau)x(\tau-T)d\tau \\
= &  x(\bar{t}_n)\bar{K}_{T}(\bar{t}_n,t) + \bar{v}(\bar{t}_n,\bar{t}_n+T) \\
& - \int_{\bar{t}_n +T}^t \mathcal{K}_{T}(\tau)\bar{v}(\bar{t}_n,\tau-T)d\tau ,
\end{split}
\end{equation}
where we have used \eqref{eq:K_T_bar_in-proof_lemma1}, and the following notation:
\begin{equation}
\bar{v}(\bar{t}_n,t)=\int_{\bar{t}_n}^tv(\tau)d\tau.
\end{equation}
ii) $v(t)=0$ for $t\geq t^{'}$ where $ \bar{t}_n +T  \leq t^{'} \leq  \bar{t}_{n+1} $.  A similar reasoning yields
\begin{equation}\label{eq:solution_closed_loopsystem_interval_Ktequal_2-distrubance}
\begin{split}
x(t) = &  x(\bar{t}_n)\bar{K}_{T}(\bar{t}_n,t) + \bar{v}(\bar{t}_n,\bar{t}_n+T) + \bar{v}(\bar{t}_n+T,t)  \\
& - \int_{\bar{t}_n +T}^t \mathcal{K}_{T}(\tau)\bar{v}(\bar{t}_n,\tau-T)d\tau ,
\end{split}
\end{equation}
where
\begin{equation}
\bar{v}(\bar{t}_n+T,t)=\int_{\bar{t}_n+T}^tv(\tau)d\tau.
\end{equation}
Thus, at $t=\bar{t}_{n+1}$ and by virtue of  \eqref{eq:Property_integral_K_T}, it holds from either  \eqref{eq:solution_closed_loopsystem_interval_Ktequal_1-distrubance}  or \eqref{eq:solution_closed_loopsystem_interval_Ktequal_2-distrubance} that $x(\bar{t}_{n+1}) \neq 0$, hence it becomes the new initial condition.
For the next interval of time $[\bar{t}_{n+1},\bar{t}_{n+1}+T]$, we have that the solution  to \eqref{eq:single_integrator_closed-loop_in_proof} is given as follows:
\begin{equation}\label{eq:solution_closed_loopsystem_interval_Ktequal_22-distrubance}
x(t)=x(\bar{t}_{n+1}),
\end{equation}
thus $x(\bar{t}_{n+1} +T)=x(\bar{t}_{n+1})$; and for $t \in [\bar{t}_{n+1}+T,\bar{t}_{n+2}]$, we obtain 
\begin{equation}\label{eq:solution_closed_loopsystem_interval_Ktequal_3-distrubance}
x(t)=x(\bar{t}_{n+1})\mathcal{K}_{T}(\bar{t}_{n+1},t),
\end{equation}
where $\bar{K}_{T}(\bar{t}_{n+1},t)$ is defined as follows (similarly to \eqref{eq:K_T_bar_in-proof_lemma1}): 
\begin{equation}\label{eq:K_T_bar_in-proof2_lemma1}
\bar{K}_{T}(\bar{t}_{n+1},t)=1 - \int_{\bar{t}_{n+1} +T}^t \mathcal{K}_{T}(\tau)d\tau,
\end{equation}
with the property $\bar{K}_{T}(\bar{t}_{n+1},\bar{t}_{n+2})=0$. Therefore, we finally obtain
\begin{equation}\label{eq:dead_beat_property_tn2}
x(\bar{t}_{n+2}) =0.
\end{equation}
Repeating the same reasoning on the interval  $ [\bar{t}_{n+2}, \bar{t}_{n+3}]$ (the control features zero gain,  acting on the zero solution), we conclude that $x(t)=0$ for all $t \geq  \bar{t}_{n+2} $.  
  Recalling that  $\bar{t}_{n+2}=2T(n+2) -\delta_{\zeta} T$,  the fact that $t^{'} \leq \bar{t}_{n+1} $,  and by the definition of the ceiling function, we have that $\frac{t^{'}+ \delta_{\zeta} T}{2T} \leq \Big \lceil \frac{t^{'}+ \delta_{\zeta} T}{2T}\Big\rceil = (n+1)$, hence the dead-beat property is achieved at $2T\lceil \frac{t^{'}+ \delta_{\zeta} T}{2T}\rceil + 2T-\delta_{\zeta} T$. By induction, it holds true for any interval,  i.e., $ \bar{t}_{n} \leq t^{'} \leq  \bar{t}_{n+1}$, for arbitrary  $n \in \mathbb{Z}^+$. Finally, estimate \eqref{eq:bound_PrT_single_integrator-control-vanishing_dist} can be established using  \eqref{eq:exponential_estimate_withGronwall}, \eqref{eq:solution_closed_loopsystem_interval_Ktequal_3-distrubance} and   \eqref{eq:dead_beat_property_tn2}.
This concludes the proof. 
 \end{proof}

\subsection{Proof of Claim \ref{claim1}}\label{proof_claim}
We follow similar arguments as in the proof of Lemma \ref{Lemma_PrT_Pointwise_delay} and divide the proof in two steps.  \textit{Step 1:}  Consider  the solutions of \eqref{eq:Newton_ES_static_maps1_prescribed-pointwise_error_averaged_target} and \eqref{eq:Newton_ES_static_maps2_prescribed-pointwise_error_averaged_target} over   $[\bar{t}_0, \bar{t}_0+T]$ where $\bar{t}_0=-\delta_{\zeta}T$, $\delta_{\zeta} \in (-1,1]$, and  recall that $T>D$. On this interval,   $\mathcal{K}_{T}(t)= 0$ and $\mathcal{L}_{T}(t)= 0$, thus $x_2(t)=x_2(\bar{t}_0)$ and $x_1(t)=x_1(\bar{t}_0)+x_2(\bar{t}_0)(t-\bar{t}_0)$. Then, it holds 
\begin{equation}\label{first_bound_Claim1_in_proof}
\sup_{-D \leq s \leq 0}(\vert x_1(t+s) \vert) \leq \vert x_1(\bar{t}_0) \vert + \vert x_2(\bar{t}_0) \vert (T-D), 
\end{equation}
for all $t\in [\bar{t}_0,\bar{t}_0+T]$. \textit{Step 2:} For $t\in[\bar{t}_0+T,\bar{t}_1]$ where $\bar{t}_1=2T-\delta_{\zeta}T$,  the   solutions of \eqref{eq:Newton_ES_static_maps1_prescribed-pointwise_error_averaged_target} and \eqref{eq:Newton_ES_static_maps2_prescribed-pointwise_error_averaged_target} satisfy 
\begin{equation}\label{eq:solution_x1_active_interval}
\begin{split}
x_1(t)=&x_1(\bar{t}_0+T)  - \int_{\bar{t}_0+T}^{t}\left(\mathcal{K}_{T}(s)x_1(s-T)ds - x_2(s)\right)ds,
\end{split}
\end{equation}
\begin{equation}\label{eq:solution_x2_active_interval}
\begin{split}
x_2(t)=&x_2(\bar{t}_0+T)  - \int_{\bar{t}_0+T}^{t}\mathcal{K}_{T}(s)x_2(s-T)ds \\
 &+ \int_{\bar{t}_0+T}^{t}\Phi(s)x_{1}(s-D)ds,
\end{split}
\end{equation}
where
\begin{equation}\label{eq:Phi_in_proof_Claim1}
 \Phi(s)=- \mathcal{L}_{T}(s)\tilde{\Gamma}_{\rm av}(s)H^{\star}.
\end{equation} 
 Using Step 1, equations  \eqref{eq:solution_x1_active_interval} and \eqref{eq:solution_x2_active_interval} simplify to: 
\begin{equation}\label{eq:solution_x1_active_intervalV2}
\begin{split}
x_1(t)=&\bar{K}_{T}(\bar{t}_0,t)x_1(\bar{t}_0) + \int_{\bar{t}_0+T}^{t}x_2(s)ds \\
&+ \Big[T-\int_{\bar{t}_0+T}^{t}\mathcal{K}_{T}(s)(s-(\bar{t}_0+T))ds\Big]x_2(\bar{t}_0), 
\end{split}
\end{equation}
\begin{equation}\label{eq:solution_x2_active_intervalV2}
\begin{split}
x_2(t)=&\bar{K}_{T}(\bar{t}_0,t)x_2(\bar{t}_0) + \int_{\bar{t}_0+T}^{t}\Phi(s)x_{1}(s-D)ds,
\end{split}
\end{equation}
for all $t\in [\bar{t}_0+T, \bar{t}_1 ]$, where $\bar{K}_{T}(\bar{t}_0,t)$ is  as in \eqref{eq:K_T_bar_in-proof_lemma1}.   
 Substituting \eqref{eq:solution_x2_active_intervalV2} into \eqref{eq:solution_x1_active_intervalV2},  and applying Fubini's theorem,  we obtain 
 \begin{equation*}\label{eq:solution_x1_active_intervalV3}
 \begin{split}
x_1(t)=& x_1(\bar{t}_0) \bar{K}_{T}(\bar{t}_0,t)  + \int_{\bar{t}_0+T}^{t}(t-s)\Phi(s)x_1(s-D)ds \\
& + x_2(\bar{t}_0)\Big((t-\bar{t}_0)\bar{K}_{T}(\bar{t}_0,t) + T\int_{\bar{t}_0+T}^{t}\mathcal{K}_{T}(s)ds  \Big),
\end{split}
\end{equation*}
for all $t \in [\bar{t}_0+T,\bar{t}_1]$. By distinguishing two cases:
i) $4T-\delta_{\xi} 2T \leq \bar{t}_1 $ or  ii)    $\bar{t}_1 \leq  4T-\delta_{\xi} 2T $, and using the boundedness of $\mathcal{K}_{T}$,   we obtain the following estimate,  for all $t\in [\bar{t}_0+T,\max\{4T-\delta_{\xi} 2T,\bar{t}_1\}]$:  
\begin{equation}
\begin{split}
\sup_{-D \leq s \leq 0}(\vert & x_1(t+s) \vert) \leq  \bar{C}(T) \Big(\vert x_1(\bar{t}_0) \vert + \vert x_2(\bar{t}_0) \vert \Big)  \\
& + \int_{\bar{t}_0+T}^{t} \vert (t-s)\vert \vert \Phi(s) \vert \sup_{-D \leq r \leq 0}(\vert x_1(s+r) \vert) ds, \\
\end{split}
\end{equation}
for some appropriate ${\bar C}>0$ depending on $T$. Then,  by the Grönwall-Bellman inequality, we  obtain 
\begin{equation}\label{eq:final_estimate_supremum_x1_Gronwall_inequ}
\begin{split}
&\sup_{-D \leq s \leq 0}(\vert x_1(t+s) \vert) \leq  \bar{C}(T) \Big(\vert x_1(\bar{t}_0) \vert + \vert x_2(\bar{t}_0) \vert \Big)  \\
& \hskip 2.2cm  \times \exp\Bigg( \int_{\bar{t}_0+T}^{\hat{t}} \vert (\hat{t}-s) \vert  \vert \Phi(s)\vert ds \Bigg), \\
\end{split}
\end{equation}
where $\hat{t}=\max\{4T-\delta_{\xi} 2T,\bar{t}_1\}$.
 The integral in  \eqref{eq:final_estimate_supremum_x1_Gronwall_inequ} is   finite since  $\Phi$, given in \eqref{eq:Phi_in_proof_Claim1},   is bounded, namely    $\vert \mathcal{L}_T(t) \vert \leq \frac{4 \pi}{T^2} $ and   $ \vert \Gamma_{\rm av}(t) \vert $ is upper bounded for all $t \in [0,\hat{t}] $  by virtue of \eqref{eq:PrT_Hessian_estimator_average_in_proof} and the relation \eqref{eq:relation_Gamma_H}.   Finally, in the case  $4T-\delta_{\xi} 2T \geq \bar{t}_1$, we can refine the analysis on  $ [\bar{t}_1,4T-\delta_{\xi} 2T]$,  by repeating the two steps as before and obtaining an estimate similar as \eqref{eq:final_estimate_supremum_x1_Gronwall_inequ}; with $x_1(\bar{t}_1)  $, $x_2(\bar{t}_1) $ as the initial conditions.
Notice, in addition,  that if $\bar{t}_0= - \delta_{\zeta} T > 0$ (e.g., when $\delta_{\zeta} <0$), it suffices to repeat the same reasoning as before, with Step 2 applied first instead of Step 1, and considering $x_1(0)  $, $x_2(0) $ as the initial conditions. 
Hence, combining \eqref{first_bound_Claim1_in_proof}, \eqref{eq:final_estimate_supremum_x1_Gronwall_inequ}  and  by the above considerations, we conclude the proof of Claim \ref{claim1}.





\bibliographystyle{IEEEtranS}
\bibliography{biblio_PT_Extremum_seeking2025}

\end{document}